\documentclass[11pt, dvipsnames]{article}

\usepackage[dvipsnames]{xcolor}

\usepackage{natbib}

\usepackage{amsthm}
\usepackage{amsmath}
\usepackage{amssymb}
\usepackage{refcount}
\usepackage{fullpage}
\usepackage{paralist}
\usepackage{appendix}
\usepackage{graphicx}
\usepackage{bm}
\usepackage{bbold}

\usepackage[ruled,boxed,vlined]{algorithm2e}
\DontPrintSemicolon
\usepackage{verbatim}

\renewcommand{\phi}{\varphi}

\newcommand{\eps}{\epsilon}

\newcommand{\hide}[1]{ }

\renewcommand{\mathbf}{\bm}

\theoremstyle{plain}

\renewcommand{\include}{\input}
\usepackage{tikz}
\usetikzlibrary{decorations.pathreplacing,angles,quotes}
\usetikzlibrary{arrows}
\usepackage{float}
\usetikzlibrary{patterns}

\newcommand{\opt}[2]{\left.OPT(#1) \right|_{#2}}
\newcommand{\optsin}[1]{OPT(#1)}
\newcommand{\solsin}[1]{SOL(#1)}

\newcommand{\alg}[2]{\left.ALG(#1) \right|_{#2}}
\newcommand{\sol}[2]{\left.SOL(#1) \right|_{#2}}

\newcommand{\conslaalphai}{1 - \frac{5\sqrt{\rho_i}}{\eps^2}}

\usepackage{libertine}

\usepackage{geometry}
\usepackage{enumitem}
\PassOptionsToPackage{vlined, boxruled}{algorithm2e}
\usepackage{algorithm2e}
\usepackage{mleftright}
\usepackage{relsize}

\usepackage{caption}
\usepackage{subcaption}

\usepackage{amsmath}
\usepackage{amsthm}
\usepackage{amssymb}
\usepackage{graphicx}
\usepackage{tabularx}
\newcolumntype{Y}{>{\centering\arraybackslash}X}

\usepackage{array}
\usepackage{soul}

\usepackage{mdframed}
\usepackage{xparse}
\usepackage{hyperref}
\hypersetup{colorlinks,linkcolor=blue,citecolor=blue,urlcolor=blue}

\usepackage[capitalise, noabbrev]{cleveref}

\usepackage{tcolorbox}
\usepackage{tikz,lipsum}
\usepackage{newtxmath}
\tcbuselibrary{skins,breakable}

\theoremstyle{plain}
\newtheorem{thm1}{Theorem}[section]
\theoremstyle{remark}
\newtheorem{pthm1}[thm1]{Theorem}
\theoremstyle{plain}
\newtheorem{lem1}[thm1]{Lemma}
\theoremstyle{plain}
\newtheorem{obs1}[thm1]{Observation}
\theoremstyle{plain}
\newtheorem{inv1}[thm1]{Invariant}
\theoremstyle{plain}
\newtheorem{cor1}[thm1]{Corollary}
\theoremstyle{definition}
\newtheorem{defn1}[thm1]{Definition}
\theoremstyle{plain}
\newtheorem{fact1}[thm1]{Fact}
\theoremstyle{remark}
\newtheorem{rem1}[thm1]{Remark}
\theoremstyle{plain}
\newtheorem{prop1}[thm1]{Proposition}
\theoremstyle{plain}
\newtheorem{asmp1}[thm1]{Assumption}

\ifx\proof\undefined
\newenvironment{proof}[1][\protect\proofname]{\par
\normalfont\topsep6\p@\@plus6\p@\relax
\trivlist
\itemindent\parindent
\item[\hskip\labelsep\scshape #1]\ignorespaces
}{%
\endtrivlist\@endpefalse
}
\providecommand{\proofname}{Proof}
\fi

\def\special{1}
\def\specialproof{1}
\def\specialdefinition{1}
\def\specialremark{1}
\def\highlight{0}
\def\sidebar{0}

\def\colorequations{0}

\newenvironment{theorem}[1][]{%
\begin{thm1}[#1]%
}{\end{thm1}%
}
\newenvironment{lemma}[1][]{%
\begin{lem1}[#1]%
}{\end{lem1}%
}
\newenvironment{observation}[1][]{%
\begin{obs1}[#1]%
}{\end{obs1}%
}
\newenvironment{inv}[1][]{%
\begin{inv1}[#1]%
}{\end{inv1}%
}
\newenvironment{fact}[1][]{%
\begin{fact1}[#1]%
}{\end{fact1}%
}
\newenvironment{remark}[1][]{%
\begin{rem1}[#1]%
}{\end{rem1}%
}
\newenvironment{pthm}[1][]{%
\begin{pthm1}[#1]%
}{\end{pthm1}%
}
\newenvironment{corollary}[1][]{%
\begin{cor1}[#1]%
}{\end{cor1}%
}
\newenvironment{definition}[1][]{%
\begin{defn1}[#1]%
}{\end{defn1}%
}
\newenvironment{asmp}[1][]{%
\begin{asmp1}[#1]%
}{\end{asmp1}%
}
\newenvironment{proposition}[1][]{%
\begin{prop1}[#1]%
}{\end{prop1}
}%

\if 1
    \if 1\highlight
        \renewenvironment{theorem}[1][]{%
        \begin{mdframed}[nobreak=false,backgroundcolor=Aquamarine!60]\begin{thm1}[#1]%
        }{\end{thm1}\end{mdframed}%
        }
        \renewenvironment{lemma}[1][]{%
        \begin{mdframed}[nobreak=false,backgroundcolor=YellowGreen!60]\begin{lem1}[#1]%
        }{\end{lem1}\end{mdframed}%
        }
        \renewenvironment{observation}[1][]{%
        \begin{mdframed}[nobreak=false,backgroundcolor=Salmon!60]\begin{obs1}[#1]%
        }{\end{obs1}\end{mdframed}%
        }
        
        \renewenvironment{fact}[1][]{%
        \begin{mdframed}[nobreak=false,backgroundcolor=Salmon!60]\begin{fact1}[#1]%
        }{\end{fact1}\end{mdframed}%
        }
        
        \renewenvironment{proposition}[1][]{%
        \begin{mdframed}[backgroundcolor=Goldenrod!60]\begin{prop1}[#1]%
        }{\end{prop1}\end{mdframed}%
        }
        
        \renewenvironment{corollary}[1][]{%
        \begin{mdframed}[backgroundcolor=Mulberry!60]\begin{cor1}[#1]%
        }{\end{cor1}\end{mdframed}%
        }
    \fi
    \if 1\sidebar
        \tcolorboxenvironment{theorem}{blanker,breakable,left=4mm, before skip=10pt,after skip=10pt, borderline west={2mm}{0pt}{Aquamarine}}
        \tcolorboxenvironment{lemma}{blanker,breakable,left=4mm, before skip=10pt,after skip=10pt, borderline west={2mm}{0pt}{YellowGreen}}
        \tcolorboxenvironment{observation}{blanker,breakable,left=4mm, before skip=10pt,after skip=10pt, borderline west={2mm}{0pt}{Salmon}}
        \tcolorboxenvironment{inv}{blanker,breakable,left=4mm, before skip=10pt,after skip=10pt, borderline west={2mm}{0pt}{Salmon}}
        \tcolorboxenvironment{fact}{blanker,breakable,left=4mm, before skip=10pt,after skip=10pt, borderline west={2mm}{0pt}{Salmon}}
        \tcolorboxenvironment{pthm}{blanker,breakable,left=4mm, before skip=10pt,after skip=10pt, borderline west={2mm}{0pt}{Salmon}}
        \tcolorboxenvironment{proposition}{blanker,breakable,left=4mm, before skip=10pt,after skip=10pt, borderline west={2mm}{0pt}{Goldenrod}}
        \tcolorboxenvironment{corollary}{blanker,breakable,left=4mm, before skip=10pt,after skip=10pt, borderline west={2mm}{0pt}{Mulberry}}
        \tcolorboxenvironment{asmp}{blanker,breakable,left=4mm, before skip=10pt,after skip=10pt, borderline west={2mm}{0pt}{Goldenrod}}
    \fi
\fi

\if 1\specialproof
    \if 1\highlight
        \expandafter\let\expandafter\oldproof\csname\string\proof\endcsname
        \let\oldendproof\endproof
        \renewenvironment{proof}[1][\proofname]{%
        \begin{mdframed}[nobreak=false,backgroundcolor=lightgray!60]\oldproof[#1]%
        }{\oldendproof\end{mdframed}}
    \else
        \if 1\sidebar
        \tcolorboxenvironment{proof}{
        blanker,breakable,left=4mm,
        before skip=10pt,after skip=10pt,
        borderline west={2mm}{0pt}{lightgray}}
        \fi
    \fi
\fi

\if 1\specialdefinition
    \if 1\highlight
        \renewenvironment{definition}[1][]{%
        \begin{mdframed}[innerbottommargin=0.1cm,innertopmargin=0.1cm,backgroundcolor=Apricot!60]\begin{defn1}[#1]%
        }{\end{defn1}\end{mdframed}%
        }
    \else
        \if 1\sidebar
        \tcolorboxenvironment{definition}{
        blanker,breakable,left=4mm,
        before skip=10pt,after skip=10pt,
        borderline west={2mm}{0pt}{Apricot}}
        \fi
    \fi
\fi

\if 1\specialremark
    \if 1\highlight
        \renewenvironment{remark}[1][]{%
        \begin{mdframed}[backgroundcolor=Salmon!60]\begin{rem1}[#1]%
        }{\end{rem1}\end{mdframed}%
        }
    \else
        \if 1\sidebar
        \tcolorboxenvironment{remark}{
        blanker,breakable,left=4mm,
        before skip=10pt,after skip=10pt,
        borderline west={2mm}{0pt}{Salmon}}
        \fi
    \fi
\fi

\if 1\colorequations
    \everymath{\color{MidnightBlue}}
\fi

\title{Multi Layer Peeling for Linear Arrangement and Hierarchical Clustering}
\author{Yossi Azar%
\thanks{School of Computer Science, Tel-Aviv University. Email: azar@tau.ac.il. Supported in part by the Israel Science Foundation (grant No. 2304/20).}
\and
Danny Vainstein%
\thanks{School of Computer Science, Tel-Aviv University. Email: dannyvainstein@gmail.com.}
}

\begin{document}

\maketitle

\begin{abstract}
We present a new multi-layer peeling technique to cluster points in a metric space. A well-known non-parametric objective is to embed the metric space into a simpler structured metric space such as a line (i.e., Linear Arrangement) or a binary tree (i.e., Hierarchical Clustering). Points which are close in the metric space should be mapped to close points/leaves in the line/tree; similarly, points which are far in the metric space should be far in the line or on the tree. In particular we consider the Maximum Linear Arrangement problem \cite{Approximation_algorithms_for_maximum_linear_arrangement} and the Maximum Hierarchical Clustering problem \cite{Hierarchical_Clustering:_Objective_Functions_and_Algorithms} applied to metrics.

We design approximation schemes ($1 - \epsilon$ approximation for any constant $\epsilon > 0$) for these objectives. In particular this shows that by considering metrics one may significantly improve former approximations ($0.5$ for Max Linear Arrangement and $0.74$ for Max Hierarchical Clustering). Our main technique, which is called multi-layer peeling, consists of recursively peeling off points which are far from the "core" of the metric space. The recursion ends once the core becomes a sufficiently densely weighted metric space (i.e. the average distance is at least a constant times the diameter) or once it becomes negligible with respect to its inner contribution to the objective. Interestingly, the algorithm in the Linear Arrangement case is much more involved than that in the Hierarchical Clustering case, and uses a significantly more delicate peeling.
\end{abstract}

\section{Introduction}

Unsupervised learning plays a major role in the field of machine learning. Arguably the most prominent type of unsupervised learning is done through clustering. Abstractly, in this setting we are given a set of data points with some notion of pairwise relations which is captured via a metric space (such that closer points are more similar). In order to better understand the data, the goal is to embed this space into a simpler structured space while preserving the original pairwise relationships. A prevalent solution in this domain is to build a flat clustering (or partition) of the data (e.g., by using the k-means algorithm). However, these types of solutions ultimately fail to capture all pairwise relations (e.g., intra-cluster relations). To overcome this difficulty, often the metric space is mapped to structures that may capture all pairwise relations - in our case into a Linear Arrangement (LA) or a Hierarchical Clustering (HC).

The idea of embedding spaces by using a Linear Arrangement or Hierarchical Clustering structure is not new. These types of solutions have been extensively used in practice (e.g., see \cite{Batch_Active_Learning_at_Scale, Scaling_Hierarchical_Agglomerative_Clustering_to_Billion_sized_Datasets, Distributed_Balanced_Partitioning_via_Linear_Embedding, Affinity_Clustering_Hierarchical_Clustering_at_Scale, Hierarchical_Clustering_of_Data_Streams_Scalable_Algorithms_and_Approximation_Guarantees}) and have also been extensively researched from a theoretical point of view (e.g., see \cite{a_cost_function_for_similarity-based_hierarchical_clustering, Hierarchical_Clustering:_Objective_Functions_and_Algorithms, Approximation_Bounds_for_Hierarchical_Clustering:_Average_Linkage, spreading_metrics_for_vertex_ordering_problems, An_improved_approximation_ratio_for_the_minimum_linear_arrangement_problem, Approximation_algorithms_for_maximum_linear_arrangement}). Notably, the Linear Arrangement type objectives were first considered by Hansen \cite{Hansen89} who considered the embedding of graphs into 2-dimensional and higher planes. On the other hand, the study of Hierarchical Clustering type objectives was initiated by Dasgupta \cite{a_cost_function_for_similarity-based_hierarchical_clustering} - spurring a fruitful line of work resulting in many novel algorithms. In practice, more often than not, the data considered adheres to the triangle inequality (in particular guaranteeing that if point $a$ is similar, equivalently close, to points $b$ and $c$ then so are $b$ and $c$) and thus may be captured by a metric (e.g., see \cite{Hierarchical_Clustering_for_Euclidean_Data,Objective_Based_Hierarchical_Clustering_of_Deep_Embedding_Vectors,Hierarchical_Clustering_of_Data_Streams_Scalable_Algorithms_and_Approximation_Guarantees})

\noindent The first objective we consider is the Max Linear Arrangement objective. 

\begin{definition}
Let $G = (V,w)$ denote a metric (specifically, $w$ satisfies the triangle inequality) with $|V| = n$. In the \textbf{Max Linear Arrangement problem} our goal is to return a 1-1 mapping $y : V \rightarrow [n]$ so as to maximize $\sum_{i,j} w_{i,j} y_{i,j}$, where $y_{i,j} = |y_i - y_j|$. 
\end{definition}

\noindent The second objective we consider is the Max Hierarchical Clustering objective.

\begin{definition}
Let $G = (V,w)$ denote a metric (specifically, $w$ satisfies the triangle inequality). In the \textbf{Max Hierarchical Clustering problem} our goal is to return a binary HC tree $T$ such that its leaves are in a 1-1 correspondence with $V$. Furthermore, we would like to return $T$ so as to maximize $\sum_{i,j} w_{i,j} |T_{i,j}|$, where $T_{ij}$ is the subtree rooted at the lowest-common-ancestor of the leaves $i$ and $j$ in the Hierarchical Clustering tree $T$ and $|T_{i,j}|$ is the number of leaves in $T_{i,j}$.
\end{definition}

% Formally, we consider the setting where we are given $n$ points in a metric $G = (V,w)$ (i.e., where $w$ satisfies the triangle inequality). The first objective we consider is the Max Linear Arrangement objective. In this setting our goal is to return a 1-1 mapping $y : V \rightarrow [n]$ so as to maximize 
% \[
% \tag{\texttt{LA}}
% \label{la_objective}
% \sum_{i,j} w_{i,j} |y_i - y_j|.
% \]
% It will be convenient for us to denote $|y_i - y_j| = y_{i,j}$. The second objective we consider is the Max Hierarchical Clustering objective. 
% In this setting our goal is to return a binary Hierarchical Clustering tree $T$ such that its leaves are in a 1-1 correspondence with $V$. Furthermore, we would like to return $T$ so as to maximize
% \[
% \tag{\texttt{HC}}
% \label{hc_objective}
% \sum_{i,j} w_{i,j} |T_{i,j}|,
% \]
% where $T_{ij}$ is the subtree rooted at the lowest-common-ancestor of the leaves $i$ and $j$ in the Hierarchical Clustering tree $T$ and $|T_{i,j}|$ is the number of leaves in $T_{i,j}$.

These objectives were first considered by Hassin and Rubinstein \cite{Approximation_algorithms_for_maximum_linear_arrangement} and  Cohen-Addad et al. \cite{Hierarchical_Clustering:_Objective_Functions_and_Algorithms} (respectively) with respect to the non-metric case. For these (non-metric) objectives the best known approximation ratios are $0.5$ for the Linear Arrangement objective \cite{Approximation_algorithms_for_maximum_linear_arrangement} and $0.74$ for the Hierarchical Clustering objective \cite{Objective_Based_Hierarchical_Clustering_of_Deep_Embedding_Vectors}). The former was achieved by was achieved by bisecting the data points randomly and thereafter greedily arranging each set and the latter was achieved by approximating the Balanced Max-2-SAT problem.

% These objectives were first considered by Hassin and Rubinstein \cite{Approximation_algorithms_for_maximum_linear_arrangement} and Cohen-Addad et al. \cite{Hierarchical_Clustering:_Objective_Functions_and_Algorithms} (respectively) with respect to the non-metric case. For these (non-metric) objectives the best known approximation ratios are $0.5$ for the Linear Arrangement objective \cite{Approximation_algorithms_for_maximum_linear_arrangement} and $0.74$ for the Hierarchical Clustering objective \cite{Objective_Based_Hierarchical_Clustering_of_Deep_Embedding_Vectors}). The former was achieved by bisecting the data points randomly and thereafter greedily arranging each set, while the latter was achieved by approximating the Balanced Max-2-SAT problem.

% In practice, more often than not, the data considered adheres to the triangle inequality (in particular guaranteeing that if point $a$ is similar, equivalently close, to points $b$ and $c$ then so are $b$ and $c$) and thus may be captured by a metric (e.g., see \cite{Hierarchical_Clustering_for_Euclidean_Data,Objective_Based_Hierarchical_Clustering_of_Deep_Embedding_Vectors,Hierarchical_Clustering_of_Data_Streams_Scalable_Algorithms_and_Approximation_Guarantees}). 

As stated earlier, more often than not, the data considered in practical applications adheres to the triangle inequality. Therefore, our results' merits are two fold. First, we offer a generalized framework to tackle these types of embedding objectives. Second, our results show that by applying this natural assumption we may significantly improve former best known approximations (from 0.5 (LA) and 0.74 (HC) to $1 - \epsilon$ for any constant $\epsilon > 0$).

\paragraph{Our Results.}
We provide the following results.
\begin{itemize}
    \item We design a general framework in order to tackle the embedding of metric spaces into simpler structured spaces (see Algorithm \ref{alg.general_alg}). We then concretely apply our framework to both the Linear Arrangement and Hierarchical Clustering settings. For an extended discussion see Our Techniques.
    
    \item We apply our framework to the Linear Arrangement case. In this case we prove that our applied algorithm (\ref{alg.la_alg}) is an EPRAS (see Definition \ref{def.pras}) - i.e., for any constant $\epsilon > 0$ it yields a $1 - \epsilon$ approximation. 

    \item We apply our framework to the Hierarchical Clustering case. In this case we prove that our applied algorithm (\ref{alg.hc_alg}) is an EPRAS (see Definition \ref{def.pras}) - i.e., for any constant $\epsilon > 0$ it yields a $1 - \epsilon$ approximation.
    
\end{itemize}

\paragraph{Our Techniques.}

Our generic multi-layer peeling approach appears in Algorithm \ref{alg.general_alg}. We begin by checking whether the metric space is sufficiently densely weighted (i.e., whether the average distance is at least a constant times the diameter, or equivalently the metric's weighted density (see Definition \ref{def.density_diameter_size}) is constant). 
% More specifically, we check whether almost all distances are within a constant factor of the instance’s diameter.
If this is the case then we apply a specific algorithm that handles such instances. In the LA case we devise our own algorithm (see Algorithm \ref{alg.la_alg_dense}). Algorithm \ref{alg.la_alg_dense} leverages the General Graph Partitioning algorithm of Goldreich et al. \cite{Property_Testing_and_its_Connection_to_Learning_and_Approximation} in order to “guess” an optimal graph partition that induces an almost optimal linear arrangement. In the HC case we leverage the work of Vainstein et al. \cite{Hierarchical_Clustering_via_Sketches_and_Hierarchical_Correlation_Clustering}.

If, however, the metric is not sufficiently densely weighted, then we observe that it must contain a \textbf{core} - a subset of nodes containing almost all data points with a diameter significantly smaller than the original metric's. Our general algorithm then peels off data points \textit{far} from the core (in the LA setting) or \textit{not in} the core (in the HC setting). We then embed these peeled off points; by placing them on one of the extreme sides of the line (in the LA setting) or by arranging them in a ladder structure (in the HC case; see Definition \ref{def.ladder}). Thus, we are left with handling the core (in the HC setting) or the extended core (in the LA setting).

Once again we consider two cases - either the total weight within the (extended) core is small enough, in which case we embed the core arbitrarily. Otherwise, we recurse on the instance induced by these data points. We claim that in every recursion step the density of the (extended) core increases significantly until eventually the recursion ends either when the (extended) core is sufficiently densely weighted or the total weight within the (extended) core is small enough. 

Our proof is based on several claims. First, we consider the metric’s (extended) core compared to the peeled off layer. 
Since our algorithm embeds the two sets separately, we need to bound the resulting loss in objective value. We show that the weights within the peeled off layer contribute negligibly towards the objective while the weights between the peeled off layer and the (extended) core, contribute significantly. Hence, it makes sense then to peel off this layer in order to maximize the gain in objective value.

% Since our algorithm splits the two sets, we need to make sure how the algorithm’s objective value suffers. We show that the contribution of weights within the peeled off layer, to the objective, is small. Additionally, we show that the contribution of the weights between the peeled off layer and the core is almost as large as one may expect. Hence, it makes sense then to peel off this layer in order to maximize the gain in objective value.

While the aforementioned is enough to bound the loss in a single recursion step, it is not enough. The number of recursion steps may not be constant which, in principle, may cause a blow up of the error. Nevertheless, we show that the error in each level is bounded by a geometric sequence and hence is dominated by the error of the deepest recursion step. Consequently, we manage to upper bound the total accumulated error by a constant that we may take to be as small as we wish.

% Since our goal is to relate the algorithm’s objective value to that of the optimal solution ($OPT$), we consider optimal solutions that are restricted such that they must first peel off the outside layer. We prove that in fact this restriction does not hinder its objective value significantly. However, on its own, this is not enough. We also need to handle the accumulated errors that are incurred throughout the recursion. The number recursion steps may not be constant which, in principle, may cause a blow up of the error. Nevertheless, we show that the error in each level is bounded by a geometric sequence and hence is dominated by the error of the deepest recursion step. Consequently, we manage to upper bound the total accumulated error by a constant that we may take to be as small as we wish. 

While at large this describes our proof techniques, the algorithm and analysis of LA objective is a bit more nuanced as we will be considering 3 sets: the metric’s core, the peeled off layer, and any remaining points which together with the core are labeled as the extended core. In this case, to be able to justify peeling off a layer, we must choose the layer more aggressively. Specifically, we define this layer as points that are sufficiently far from the core (rather than any point outside the core, as in the HC case). Fortunately, this defined layer (see Algorithm \ref{alg.la_alg}) fits our criteria (of our general algorithm, Algorithm \ref{alg.general_alg}). %We then need to control the effect on the solution’s value caused by the three sets (as opposed to the simpler HC case in which we had two such sets to handle)

% While, at large this describes our proof techniques, the algorithm and analysis of LA objective is a bit more nuanced (as we will be considering 3 sets: the metric’s core, the peeled off layer, and any remaining points which together with the core are labeled as the extended core). In this case, to guarantee that the peeled off layer contributes maximally to the objective (and simultaneously not causing a large drop in $OPT$’s value following the mentioned restriction), we must define this layer as points that are far from the core (rather than any point outside the core, as in the LA case). Fortunately, this defined layer (see Algorithm \ref{alg.la_alg}) fits our criteria. We then need to control the effect on the solution’s value caused by the three sets (as opposed to the simpler HC case in which we had two such sets to handle).

\paragraph{Related Work.} 
While the concept of hierarchical clustering has been around for a long time, the HC objective is relatively recent. In their seminal work, Dasgupta \cite{a_cost_function_for_similarity-based_hierarchical_clustering} considered the problem of HC from an optimization view point. Thereafter, Cohen-Addad et al. \cite{Hierarchical_Clustering:_Objective_Functions_and_Algorithms} were the first to consider the objective we use in our manuscript. In their work they showed that the well known Average-Linkage algorithm yields an approximation of $\frac23$. Subsequently, Charikar et al. \cite{Hierarchical_Clustering_better_than_Average_Linkage} improved upon this result through the use of semidefinite programming - resulting in a 0.6671 approximation. Finally, Naumov et al. \cite{Objective_Based_Hierarchical_Clustering_of_Deep_Embedding_Vectors} improved this to 0.74 by approximating the Balanced Max-2-SAT problem. 
With respect to the Max LA objective,  Hassin and Rubinstein \cite{Approximation_algorithms_for_maximum_linear_arrangement} were first to consider the problem. Through an approach of bisection and then greedily arranging the points, Hassin and Rubinstein managed to achieve a $0.5$ approximation. We note that the previous mentioned results all hold for arbitrary weights, while our main contribution is showing that by assuming the triangle inequality (i.e., metric-based dissimilarity weights) we may achieve PTAS's for both objectives. We further note that with respect to metric-based dissimilarity weights, specifically an L1 metric, Rajagopalan et al. \cite{Hierarchical_Clustering_of_Data_Streams_Scalable_Algorithms_and_Approximation_Guarantees} proved a 0.9 approximation through the use of random cut trees.

Both objectives have been originally studied with respect to their minimization variants. The minimum LA setting was first considered by Hansen \cite{Hansen89}. Hansen leveraged the work of  Leighton and Rao \cite{Multicommodity_max_flow_min_cut_theorems_and_their_use_in_designing_approximation_algorithms} on balanced separators in order to approximate the minimum linear arrangement objective to facor of $O(\log^2 n)$. Following several works improving upon this result, both Charikar et al. \cite{spreading_metrics_for_vertex_ordering_problems} and Feige and Lee \cite{An_improved_approximation_ratio_for_the_minimum_linear_arrangement_problem} leveraged the novel work of Arora et al. \cite{Expander_flows_geometric_embeddings_and_graph_partitioning} on rounding of semidefinite programs, and combined this with the rounding algorithm of Rao and Reicha \cite{New_Approximation_Techniques_for_Some_Ordering_Problems} in order to show a $O(\sqrt{\log n} \log \log n)$ approximation. For further reading on these are related types of objectives see \cite{Divide_and_Conquer_Approximation_Algorithms_via_Spreading_Metrics, New_Approximation_Techniques_for_Some_Ordering_Problems,Packing_Directed_Circuits_Fractionally, Ordering_Problems_Approximated_Single_Processor_Scheduling_and_Interval_Graph_Completion}.  On the other hand, as mentioned earlier the minimum HC setting was introduced by Dasgupta \cite{a_cost_function_for_similarity-based_hierarchical_clustering} and extensively studied as well (e.g., see \cite{a_cost_function_for_similarity-based_hierarchical_clustering, Hierarchical_Clustering:_Objective_Functions_and_Algorithms, Approximate_Hierarchical_Clustering_via_Sparsest_Cut_and_Spreading_Metrics, Hierarchical_Clustering_better_than_Average_Linkage, Bisect_and_Conquer:_Hierarchical_Clustering_via_Max-Uncut_Bisection, Hierarchical_Clustering:_a_0.585_Revenue_Approximation, Hierarchical_Clustering_via_Sketches_and_Hierarchical_Correlation_Clustering}). 

Most related to our work is that of de la Vega and Kenyon \cite{A_Randomized_Approximation_Scheme_for_Metric_MAX_CUT}. In their work they provide a PTAS for the Max Cut problem given a metric. The algorithm works by first creating a graph of clones (wherein each original vertex is cloned a number of times that is based on its outgoing weight in the original metric) with the property of being dense. It thereafter solves the problem in this new graph by applying the algorithm of de la Vega and Karpinski \cite{Polynomial_time_approximation_of_dense_weighted_instances_of_MAX_CUT}. For our objectives (HC and LA) such an approach seems to fail - specifically due to the fact that our objectives take into consideration the number of nodes in every induced cut and the cloned graph inflates the number of nodes which in turn inflates our objective values. Thus, for our considered types of objectives we need the more intricate process of iterative peeling (and subsequently terminating the process with more suited algorithms that leverage the General Graph Partitioning algorithm of Goldreich et al. \cite{Property_Testing_and_its_Connection_to_Learning_and_Approximation}).
% (and subsequently terminating the process with the more suited algorithms - Algorithm \ref{alg.la_alg_dense} (LA case) or an algorithm from \cite{Hierarchical_Clustering_via_Sketches_and_Hierarchical_Correlation_Clustering} (HC case), both using the General Graph Partitioner of Goldreich et al. \cite{Property_Testing_and_its_Connection_to_Learning_and_Approximation}). 
It is worth while mentioning that there has also been an extensive study of closely related objectives with respect to dense instances (e.g. see \cite{How_to_rank_with_few_errors,Polynomial_Time_Approximation_Schemes_for_Dense_Instances_of_NP_Hard_Problems,Linear_time_approximation_schemes_for_the_Gale_Berlekamp_game_and_related_minimization}). However these types of approaches seem to fall short since our considered metrics need not be dense.

\vspace*{-\baselineskip}

\section{Multi-Layer Peeling Framework}
\label{sec.multi_layer_peeling_framework}

% Can probably remove.
% \TODO{Throughout our paper we  consider different metric-based objectives. In order to solve them, we  apply the same recipe - if the instance is sufficiently densely weighted, apply an algorithm for these types of instances. Otherwise, the algorithm detects the metric's core (which is a small-diameter subset containing almost all nodes) and peel off (and subsequently embed) a layer of data points that are far from the core. The algorithm then considers the core; if it is sufficiently small (in terms of inner weights) then we embed the core arbitrarily and halt. Otherwise, we recurse on the core.}

% \begin{figure}
% \centering
% \begin{subfigure}{.5\textwidth}
%   \centering
%   \includegraphics[scale=0.3]{recursion_step_hc.png}
% %   \caption{A subfigure}
%   \label{fig:sub1}
% \end{subfigure}%
% \begin{subfigure}{.5\textwidth}
%   \centering
%   \includegraphics[scale=0.3]{recursion_step_la.png}
% %   \caption{A subfigure}
%   \label{fig:sub2}
% \end{subfigure}
% \caption{A recursion step (case (c)) and the two possible halting steps (cases (a) and (b)). The yellow points define the metric's core. In the HC case we peel off both red and green points in a single step, while in the LA we must be more delicate and only peel off the green points.}
% \label{fig:test}
% \end{figure}
%width=.4\linewidth

\begin{figure}
\centering
\begin{subfigure}{.5\textwidth}
  \centering
  \includegraphics[scale=0.25]{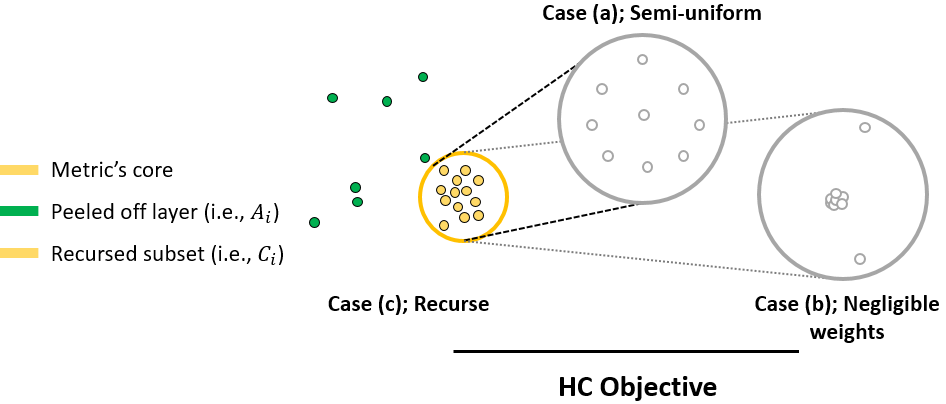}
%   \caption{A subfigure}
  \label{fig:sub1}
\end{subfigure}%
\begin{subfigure}{.5\textwidth}
  \centering
  \includegraphics[scale=0.25]{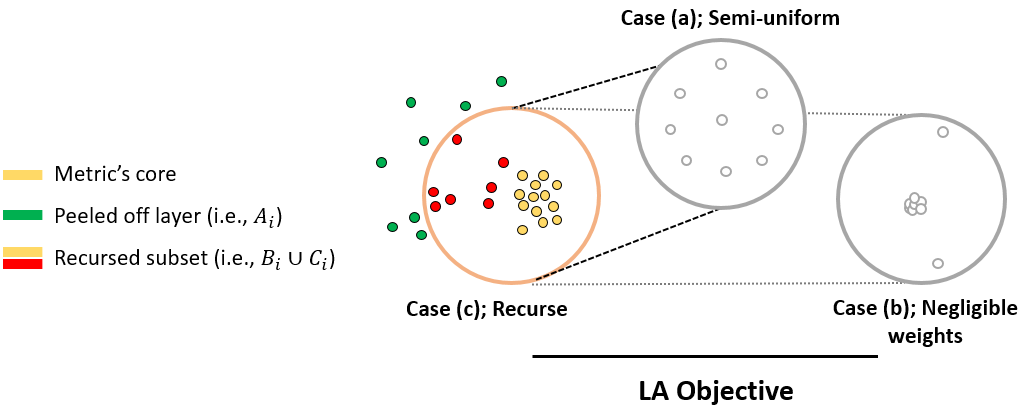}
%   \caption{A subfigure}
  \label{fig:sub2}
\end{subfigure}
\caption{A recursion step (case (c)) and the two possible halting steps (cases (a) and (b)). The yellow points define the metric's core. In the HC case we peel off both red and green points in a single step, while in the LA we must be more delicate and only peel off the green points.}
\label{fig:test}
\end{figure}

\noindent Before defining our algorithms we need the following definitions.

% \begin{definition}
% \label{def.density_diameter_size}
% Let $G = (V,w)$ denote a metric. Let $U \subset V$ denote a subset of its nodes. We introduce the following notations.
% \begin{enumerate}
%     \item Let $\mathbf{D_U} = \max_{i,j \in U} w_{i,j}$ denote $U$'s \textbf{diameter}. 
%     \item Let $\mathbf{W_U} = \sum_{i,j \in U} w_{i,j}$ denote $U$'s \textbf{sum of weights}. 
%     \item Let $\mathbf{n_U} = |U|$ denote $U$'s \textbf{size}.
%     \item Let $\mathbf{\rho_U} = \frac{W_U}{n_U^2 D_U}$ denote $U$'s \textbf{weighted density}\footnote{Typically the density is defined with respect to ${n \choose 2}$. For ease of presentation, we chose to define it with respect to $n^2$ - the proofs remain the same using the former definition.}.
% \end{enumerate}
% \end{definition}

\begin{definition}
\label{def.density_diameter_size}
Let $G = (V,w)$ denote a metric and $U \subset V$ denote a subset of its nodes. We introduce the following notations: (1) let $\mathbf{D_U} = \max_{i,j \in U} w_{i,j}$ denote $U$'s \textbf{diameter}, (2) let $\mathbf{W_U} = \sum_{i,j \in U} w_{i,j}$ denote $U$'s \textbf{sum of weights}, (3) let $\mathbf{n_U} = |U|$ denote $U$'s \textbf{size} and (4) let $\mathbf{\rho_U} = \frac{W_U}{n_U^2 D_U}$ denote $U$'s \textbf{weighted density}\footnote{Typically the density is defined with respect to ${n \choose 2}$. For ease of presentation, we chose to define it with respect to $n^2$ - the proofs remain the same using the former definition.}.
% \begin{enumerate}
%     \item Let $\mathbf{D_U} = \max_{i,j \in U} w_{i,j}$ denote $U$'s \textbf{diameter}. 
%     \item Let $\mathbf{W_U} = \sum_{i,j \in U} w_{i,j}$ denote $U$'s \textbf{sum of weights}. 
%     \item Let $\mathbf{n_U} = |U|$ denote $U$'s \textbf{size}.
%     \item Let $\mathbf{\rho_U} = \frac{W_U}{n_U^2 D_U}$ denote $U$'s \textbf{weighted density}\footnote{Typically the density is defined with respect to ${n \choose 2}$. For ease of presentation, we chose to define it with respect to $n^2$ - the proofs remain the same using the former definition.}.
% \end{enumerate}
\end{definition}

All our algorithms will make use of the following simple yet useful structural lemma that states that for small-density instances there exists a large cluster of nodes with a small diameter. The proof is deferred to the Appendix.

\begin{lemma}
\label{general.lemma.1}
For any metric $G = (V,w)$ there exists a set $U \subset V$ such that $D_U \leq 4D_V \sqrt{\rho_V}$ and $n_U \geq n_V(1 - \sqrt{\rho_V})$.
\end{lemma}

\begin{definition}
\label{def.metric_core}
Given a metric $G = (V,w)$ we denote $U \subset V$ as guaranteed by Lemma \ref{general.lemma.1} as a metric's \textbf{core}.
\end{definition}

Note that the core can be found algorithmicaly simply through brute force (while the core need not be unique, our algorithms will choose one arbitrarily). 

Throughout our paper we  consider different metric-based objectives. In order to solve them, we  apply the same recipe - if the instance is sufficiently densely weighted, apply an algorithm for these types of instances. Otherwise, the algorithm detects the metric's core (which is a small-diameter subset containing almost all nodes) and peel off (and subsequently embed) a layer of data points that are far from the core. The algorithm then considers the core; if it is sufficiently small (in terms of inner weights) then we embed the core arbitrarily and halt. Otherwise, we recurse on the core. Our algorithms for both objectives (LA and HC) will follow the same structure as defined in Algorithm \ref{alg.general_alg}.

% \begin{itemize}
%     \item (a) If the instance is dense, solve it using $ALG_{dense}$.
%     \item Otherwise, Let $C$ denote the cluster as defined by Lemma \ref{general.lemma.1}.
%     \item Find a set of outliers $A$, defined according to $C$.
%     \item Set the points in $A$.
%     \item (b) If $W_{V \setminus A}$ is negligible, merge $V \setminus A$ arbitrarily and return.
%     \item (c) Otherwise, continue recursively on $V \setminus A$.
% \end{itemize}

\begin{algorithm}[h]
\caption{General Algorithm}
\label{alg.general_alg}
\uIf(\tcp*[f]{case (a)}){the instance is sufficiently densely weighted}{
    Solve it using $ALG_{d-w}$.
}
\Else{
    Let $C$ denote the metric's core (as defined by Definition \ref{def.metric_core}).
    
    Define the layer to peel off $A \subset V \setminus C$ appropriately. 
    
    Embed $A$.
    
    % , with respect to $C$ ($C$ is not necessarily equal to $V \setminus A$).

    \lIf(\tcp*[f]{case (b)}){$W_{V \setminus A}$ is negligible}{
        Embed $V \setminus A$ arbitrarily and return.
    }
    \lElse(\tcp*[f]{case (c)}){
        Continue recursively on $V \setminus A$
    }
}
\end{algorithm}

We denote by cases (a) and (b) the different cases for which the algorithm may terminate and by case (c) the recursive step. We further denote by $ALG_{d-w}$ an auxiliary algorithm that will handle sufficiently densely weighted instances. (These algorithms will differ according to the different objectives).

Henceforth, given an algorithm $ALG$ and metric $G$ we denote by $ALG(G)$ the algorithm's returned embedding. We note that when clear from context we overload the notation and denote $ALG(G)$ as the embedding's value under the respective objectives. Equivalently, we will use the term $OPT(G)$ for the optimal embedding.

Our different algorithms will be similarly defined and thus so will their analyses. Thus, we introduce a general scheme for analyzing such algorithms. Let $k$ denote the number of recursive calls our algorithm performs. Furthermore, let $G_i$ denote the instance the algorithm is called upon in step $i$ for $i = 0, 1, \ldots, k$. (I.e., $G = G_0$ and $ALG(G_k)$ does not perform a recursive step, meaning that it terminates with case (a) or (b)). We first observe that by applying a simple averaging argument we get the following useful observation.

\begin{observation}
\label{general.ob.1}
If there exist $\alpha_i, \beta_i, \gamma_i > 0$ such that $ALG(G_i) \geq \alpha_i + ALG(G_{i+1})$ and $OPT(G_i) \leq \beta_i + \gamma_i OPT(G_{i+1})$ for all $i = 0, \ldots, k-1$ then 
\[
\frac{ALG(G)}{OPT(G)} \geq 
\frac{\sum_{i=0}^{k-1} \alpha_i + ALG(G_k)}{\sum_{i=0}^{k-1} \big(\beta_i \Pi_{j=0}^{i-1} \gamma_j \big) + (\Pi_{i=0}^{k-1}\gamma_i)OPT(G_k)} \geq
\min\{ \min_i \{\frac{\alpha_i}{\beta_i \Pi_{j=0}^{i-1} \gamma_j}\}, \frac{ALG(G_k)}{(\Pi_{i=0}^{k-1} \gamma_i)OPT(G_k)}\}.
\]
\end{observation}

Thus, in order to analyze a given algorithm, it will be enough to set the values of $\alpha_i$, $\beta_i$ and $\gamma_i$, and further analyze the approximation ratio of $\frac{ALG(G_k)}{OPT(G_k)}$ for the different terminating cases (cases (a) and (b)).

\section{Notations and Preliminaries}
\label{sec.notations_prelims}

We introduce the following notation to ease our presentation later on.

\begin{definition}
Given a metric $G=(V,w)$, a solution $\mathbf{\solsin{G}}$ for the LA objective and disjoints sets $A,B \subset V$ we define: $\mathbf{\sol{G}{A}} = \sum_{i, j \in A} w_{i,j} y_{i,j}$ and $\mathbf{\sol{G}{A,B}} = \sum_{i \in A, j \in B} w_{i,j} y_{i,j}$. For the HC objective the notations are defined symmetrically by replacing $y_{i,j}$ with $|T_{i,j}|$.
\end{definition}

\noindent We will make use of algorithms belonging to the following class of algorithms.

\begin{definition}
\label{def.pras}
An algorithm is considered an Efficient Polytime Randomized Approximation Scheme (EPRAS) if for any $\epsilon > 0$ the algorithm has expected running time of  $f(\frac{1}{\epsilon}) n^{O(1)}$ and approximates the optimal solution's value up to a factor of $1 - \epsilon$.  
\end{definition}

% \noindent While the following observations are simple, we will use them frequently and thus we state them here.

\noindent We will frequently use the following (simple) observations and thus we state them here.

\begin{observation}
\label{general.ob.0}
Given values $\alpha_i \geq 0$, $\alpha \in (0,  \frac{1}{k(k+1)})$ and $k \in \mathbb{N}$ we have: \textbf{(1)} $\Pi_{i}(1 - \alpha_i) \geq 1 - \sum_i \alpha_i$, \textbf{(2)} $1 + k \alpha < \frac{1}{1 - k \alpha} < 1 + (k+1) \alpha$ and \textbf{(3)} $1 + k \alpha < e^{k \alpha} < 1 + (k+1) \alpha$.
% \begin{equation*}
%     \begin{gathered}
%     \Pi_{i}(1 - \alpha_i) \geq 1 - \sum_i \alpha_i; \quad
%     1 + k \alpha < \frac{1}{1 - k \alpha} < 1 + (k+1) \alpha; \quad
%     1 + k \alpha < e^{k \alpha} < 1 + (k+1) \alpha
%     \end{gathered}
% \end{equation*}
% \begin{enumerate}
%     \item $\Pi_{i}(1 - \alpha_i) \geq 1 - \sum_i \alpha_i$.
%     \item $1 + k \alpha < \frac{1}{1 - k \alpha} < 1 + (k+1) \alpha.$
%     \item $1 + k \alpha < e^{k \alpha} < 1 + (k+1) \alpha.$
% \end{enumerate}
\end{observation}

% \begin{observation}
% \label{general.ob.0.1}
% For any $k \in \mathbb{N}$ and any $\alpha \in (0,  \frac{1}{k(k+1)})$ we have $1 + k \alpha < \frac{1}{1 - k \alpha} < 1 + (k+1) \alpha.$
% \end{observation}

% \begin{observation}
% \label{general.ob.0.2}
% For any $k \in \mathbb{N}$ and any $\alpha \in (0,  \frac{1}{k(k+1)})$ we have $1 + k \alpha < e^{k \alpha} < 1 + (k+1) \alpha.$
% \end{observation}

\noindent The following facts will prove useful in our subsequent proofs and are therefore stated here. 

% \begin{fact}
% \label{hc.fact.avg_link}
% Given an a metric $G$ and its optimal hierarchical clustering under the HC objective we have
% \[
% OPT(G) \geq \frac23 n \sum_{i,j} w_{i,j}.
% \]
% \end{fact}

% \begin{fact}
% \label{la.fact.avg_link}
% Given an a metric $G$ and its optimal linear arrangement under the LA objective we have
% \[
% OPT(G) \geq \frac13 n \sum_{i,j} w_{i,j}.
% \]
% \end{fact}

\begin{fact}
\label{fact.avg_link}
Given a metric $G$, if the optimal linear arrangement under the LA objective is $OPT_{LA}(G)$ and the optimal hierarchical clustering under the HC objective is $OPT_{HC}(G)$ then we have  $OPT_{LA}(G) \geq \frac13 n \sum_{i,j} w_{i,j} y_{i,j}$ and $OPT_{HC}(G) \geq \frac23 n \sum_{i,j} w_{i,j} |T_{i,j}|$.
\end{fact}

We note that the HC portion of Fact \ref{fact.avg_link} has been used widely in the literature (e.g., see proof in \cite{Hierarchical_Clustering:_Objective_Functions_and_Algorithms}). The LA portion of Fact \ref{fact.avg_link} is mentioned in Hassin and Rubinstein \cite{Approximation_algorithms_for_maximum_linear_arrangement}. Finally, in the HC section we make use of "ladder" HC trees. We define them here.

\begin{definition}
\label{def.ladder}
We define a "ladder" as an HC tree that cuts a single data point from the rest at every cut (or internal node).
\end{definition}

\section{The Linear Arrangement Objective}
\label{la.sec}

We will outline the section as follows. We begin by presenting our algorithms (first the algorithm that handles case (a) and thereafter the general algorithm). We will then bound the algorithm's approximation guarantee (by following the bounding scheme of Observation \ref{general.ob.1}). Finally, we will analyze the algorithm's running time.

\subsection{Defining the Algorithms}
\label{la.sec.defining_algs}

Here we begin by applying our general algorithm to the linear arrangement problem (which we will denote simply as $ALG$). The algorithm uses, as a subroutine, an algorithm to handle case (a). We denote this subroutine as $ALG_{d-w}$ and define it following the definition of $ALG$.

\subsubsection{Defining $ALG$}

Here we apply our general algorithm (Algorithm \ref{alg.general_alg}) to the linear arrangement setting. In order to do so, roughly speaking, we define the layer to peel off $A$ as the set of all points which are "far" from the metric's core. We also introduce a subroutine to handle densely weighted instances, $ALG_{d-w}$.\\

\begin{algorithm}[H]
\caption{Linear Arrangement Algorithm ($ALG$)}
\label{alg.la_alg}
\lIf(\tcp*[f]{case (a)}){$\rho \geq \epsilon^6$}{
    solve it using $ALG_{d-w}$.
}
\Else{
    Let $C$ denote the metric's core (as defined by Lemma \ref{general.lemma.1}).
    
    Let $A$ denote all data points that are of distance $\geq \epsilon^2 D_V$ from $C$.
    
    Place $A$ to the left of $V \setminus A$. Arrange $A$ arbitrarily.
    
    \lIf(\tcp*[f]{case (b)}){$W_{V \setminus A} < \epsilon W_V$}{
        Arrange $V \setminus A$ arbitrarily and return.
    }
    \lElse(\tcp*[f]{case (c)}){
        Continue recursively on $V \setminus A$.
    }
}
\end{algorithm}

\noindent The set $V \setminus \{A \cup C\}$ will be used frequently in the upcoming proofs and thus we give it its own notation.

\begin{definition}
Denote $B = V \setminus \{A \cup C\}$ where $A$ and $C$ are defined as in Algorithm \ref{alg.la_alg}.
\end{definition}

\subsubsection{Defining $ALG_{d-w}$}
\label{la.section.dense_algorithm}

Here we will introduce an algorithm to handle case (a) type instances. Before formally defining the algorithm, we will first provide some intuition. Towards that end we first introduce the following definition.

\begin{definition}
\label{la.def.P*}
Consider $OPT(G_k)$'s embedding into the line, $[n]$. Partition $[n]$ into $\frac{1}{\epsilon}$ consecutive sets each of size $\epsilon n$ and let $P_i^*$ denote the points embedded by $OPT(G_k)$ into the $i$'th consecutive set. Furthermore, denote by $P^* = \{P_i^*\}$ the induced partition of the metric.
\end{definition}

Later on, we will show that $OPT(G_k)$'s objective value is closely approximated by the value generated solely from inter-partition-set edges (i.e., any $(u,v)$ where $u,v$ lie in different  partition sets of $P^*$). While $OPT(G_k)$ cannot be found algorithmically, assuming the above holds, it is enough for $ALG_{d-w}$ to guess the partition $P^*$. Indeed, that is exactly what we will do, by using the general graph partitioning algorithm of Goldreich et al. \cite{Property_Testing_and_its_Connection_to_Learning_and_Approximation}.

We denote the General Graph Partitioning algorithm of Goldreich et al. \cite{Property_Testing_and_its_Connection_to_Learning_and_Approximation} as $PT(G, \Phi, \epsilon_{err})$. See Definition \ref{general.def.ggr_inequalities} for a definition of $\Phi$ and $\epsilon_{err}$ (these will be defined by $ALG_{d-w}$ as well) and see Theorem \ref{general.thm.ggr} for the tester's guarantees. We are now ready to define our algorithm that handles sufficiently densely weighted instances (Algorithm \ref{alg.la_alg_dense}).

\begin{algorithm}[H]
\caption{LA Algorithm for Sufficiently Densely Weighted Instances ($ALG_{d-w}$)}
\label{alg.la_alg_dense}
Let $k = \frac{1}{\epsilon}$ denote the size of the partition.

% Let $\lambda_i^{LB} = \lambda_i^{UB} = \epsilon \cdot n$.

\For{$\{\mu_{j,j'}\}_{j\leq k, j' \leq k, j\neq j'} \subset \{i \epsilon^9 n^2 D_V: i \in \mathbb{N} \land i \leq \frac{1}{\epsilon^7}\}$}{
    Let $\Phi = \{\epsilon n, \epsilon n \}_{j = 1}^k \cup \{\mu_{j,j'}, \mu_{j,j'} \}_{j,j' = 1}^k$.
    
    Run $PT(G, \Phi, \epsilon_{err} = \epsilon^9)$. Let $P$ denote the output partition (if succeeded).
    
    Let $\widehat{y}$ denote the linear arrangement obtained from embedding $P$ consecutively on the line (and arbitrarily within the partition sets).
    
    Compute the value $\sum_e w_e \widehat{y_e}$ for $P$.
}
Return the partition with maximum $\sum_e w_e \widehat{y_e}$ value.
\end{algorithm}

\subsection{Analyzing the Approximation Ratio of $ALG$}
\label{la.sec.analyzing_approx}

Now that we have defined $ALG$ we are ready to analyze its approximation ratio. Recall that by Observation \ref{general.ob.1} it is enough to analyze the approximation ratio of cases (a), (b) and the total loss incurred by the recursion steps (i.e., by setting $\alpha_i$, $\beta_i$ and $\gamma_i$).

\subsubsection{Structural Lemmas}

Recall that we defined $k$ to be the number of recursion steps used by $ALG$ and that $G_i$ is the instance that $ALG$ is applied to at recursion step $i$. Further recall that given $G_i$, $ALG(G_i)$ partitioned the instance into $A_i, B_i$ and $C_i$ and that, informally, by Lemma \ref{general.lemma.1} $n_{C_i}$ contains the majority of the data points and $D_{C_i}$ is relatively small compared to $D_{V_i}$. 

By the definition of $C_i$, $A_i$ could be considered as a set of outliers. Therefore, intuitively it makes sense to split $A_i$ from $C_i$. In order to prove our algorithm's approximation ratio we will show that in fact one does not lose too much compared to optimal solution, by splitting $A_i$ from $C_i$. In order to do so we will show that in fact, both the values of $ALG$ and $OPT$ will be roughly equal to $\frac12 n W_{A_i, C_i}$ (which makes sense intuitively since $C_i$ is of low diameter and contains many points and $A_i$ are the points that are far from this cluster). 

The following lemmas consider 2 types of algorithms - algorithms that split $A_i$ and $C_i$ and algorithms that do not. Furthermore, they show that in fact, by the structural properties of $A_i$ and $C_i$, if we consider the values generated by these 2 types of algorithms restricted to the objective value generated by the inter-weights $W_{A_i, C_i}$, are approximately equal. We begin by lower bounding the value generated by algorithms that split $A_i$ and $C_i$. Due to lack of space, we defer the following proofs to the Appendix.

\begin{lemma}
\label{la.lemma.6}
Given the two disjoint sets $C_i$ and $A_i$ and a linear arrangement $y$ that places all nodes in $A_i$ to the left of all nodes in $C_i$ we are guaranteed that
\[
\sum_{c \in C_i, a \in A_i} w_{a,c} y_{a,c} \geq 
\frac{n_{C_i}}{2}(W_{C_i,A_i} - n_{C_i} n_{A_i} D_{C_i}).
\]
\end{lemma}

Due to the fact that $C_i$ is a small cluster containing most of the data points the above lemma reduces to the following corollary.

\begin{corollary}
\label{la.cor.0}
Given any linear arrangement $y$ that places all nodes in $A_i$ to the left of all nodes in $C_i$ we are guaranteed that
\[
\sum_{a \in A_i, c \in C_i} w_{a,c} y_{a,c} \geq 
\frac12 n W_{A_i, C_i}(1 - \frac{5 \sqrt{\rho}}{\epsilon^2})
\]
\end{corollary}

Now that we have lower bounded algorithms that split $A_i$ and $C_i$ we will upper bound algorithms that do not have this restriction. (Note that we begin by handling the case where one of the disjoint sets is a single data point and thereafter generalize it to two disjoint sets).

\begin{lemma}
\label{la.lemma.6_5}
Given a set $C_i$ and a point $p \not \in C_i$, we are guaranteed that 
\begin{equation*}
\begin{gathered}
\sum_{c \in C_i} w_{p, c} y_{p,c} \leq (W_{p,C_i} + n_{C_i} D_{C_i})(n - \frac{n_{C_i}}{2}).
\end{gathered}
\end{equation*}
\end{lemma}

\noindent We are now ready to upper bound the inter-objective-value of two sets of disjoint points.

\begin{lemma}
\label{la.lemma.7}
Given the two disjoint sets $C_i$ and ${A_i}$ and any linear arrangement $y$ we are guaranteed that
\[
\sum_{c \in C_i,a \in {A_i}} w_{a,c} y_{a,c} \leq 
(n - \frac{n_{C_i}}{2})(W_{C_i,{A_i}} + n_{C_i} n_{A_i} D_{C_i}).
\]
\end{lemma}

Due to the fact that $C_i$ is a small cluster containing most of the data points the lemma reduces to the following corollary.

\begin{corollary}
\label{la.cor.1}
Given any linear arrangement $y$ we are guaranteed that
\[
\sum_{a \in A_i, c \in C_i} w_{a,c} y_{a,c} \leq 
\frac12 nW_{A_i,C_i}(1 + \frac{9 \sqrt{\rho}}{\epsilon^2}).
\]
\end{corollary}

We will want to show that the objective values of both $ALG$ and $OPT$  (and some other intermediate values that will be defined later on) are approximately determined by their value on the inter-weights of $W_{A_i, C_i}$. In order to do so, we first introduce the following structural lemma that will help us explain this behaviour.

\begin{lemma}
\label{la.lemma.7.5.5}
Given an instance $G$ and sets $A, B$ and $C$ as defined by $ALG(G)$ we have $W_A + W_{A,B} \leq 2\frac{\sqrt{\rho}}{\epsilon^2}W_{A, C}.$
\end{lemma}

\subsubsection{Analyzing the Approximation Ratio of Case (a) of $ALG$}

We first give an overview the approximation ratio analysis. Recall the definition of $P^*$ (Definition \ref{la.def.P*}). The first step towards our proof, is to show that instead of trying to approximate $OPT(G_k)$, it will be enough to consider its value restricted to intra-partition-set weights with respect to $P^*$. Even more, for such weights $w_{u,v}$, incident to $P^*_i$ and $P^*_{i+j}$, it will be enough to assume that their generated value towards the objective (i.e., the value $y_{u,v}$) is only $(j - 1)\epsilon n$ (while it may be as large as $(j+1)\epsilon n$). Formally, this will be done in Lemma \ref{la.lemma.1.5} (whose proof is deferred to the Appendix).

Next, recall that $ALG_{d-w}$ tries to guess the partition $P^*$ (up to some additive error) and let $P$ denote the partition guessed by $ALG_{d-w}$. Observe that if guessed correctly, the value generated towards $ALG$'s objective for any intra-partition-set weight crossing between $P_i$ and $P_{i+j}$ is at least $|P_{i+1}| + \cdots |P_{i+j - 1}|$ and if we managed to guess the set sizes as well then this value is exactly $(j - 1) \epsilon n$ (equivalent to that of $OPT$'s). This will be done in Proposition \ref{la.prop.2}.

\begin{lemma}
\label{la.lemma.1.5}
Given the balanced line partition of set sizes $\epsilon n$, denoted as $P^*$, we have
\[
OPT(G_k) \leq  (1 + 13 \epsilon) \sum_{\substack{1\leq i \leq k-1\\ 1 \leq j \leq k-i}}  W_{P^*_i, P^*_{i+j}} (|P^*_{i+1}| + \cdots + |P^*_{i+j-1}|).
\]
\end{lemma}

\noindent Before proving Proposition \ref{la.prop.2} we state the properties of the general graph partitioning algorithm of Goldreich et al. \cite{Property_Testing_and_its_Connection_to_Learning_and_Approximation}. 

\begin{definition}[\cite{Property_Testing_and_its_Connection_to_Learning_and_Approximation}]
\label{general.def.ggr_inequalities}
Let $\Phi = \{\lambda_j^{LB}, \lambda_j^{UB} \}_{j = 1}^k \cup \{\mu_{j,j'}^{LB}, \mu_{j,j'}^{UB} \}_{j,j' = 1}^k$ denote a set of non-negative values such that $\lambda_j^{LB} \leq \lambda_j^{UB}$ and $\mu_{j,j'}^{LB} \leq \mu_{j,j'}^{UB}$. We define $\mathcal{GP}_{\Phi}$ the set of graphs $G$ on $n$ vertices that have a $k$ partition $(V_1, \ldots, V_k)$ upholding the following constraints
\begin{equation*}
    \begin{gathered}
        \forall j: \, \lambda_j^{LB} \leq \frac{|V_j|}{n} \leq \lambda_j^{UB}; \quad
        \forall j, j': \, \mu_{j,j'}^{LB} \leq \frac{W_{V_j, V_{j'}}}{n^2} \leq \mu_{j,j'}^{UB}.
    \end{gathered}
\end{equation*}
% \begin{itemize}
%     \item $\forall j: \,\, \lambda_j^{LB} \leq \frac{|V_j|}{n} \leq \lambda_j^{UB}$,
%     \item $\forall j, j': \,\, \mu_{j,j'}^{LB} \leq \frac{W_{V_j, V_{j'}}}{n^2} \leq \mu_{j,j'}^{UB}$.
% \end{itemize}
\end{definition}

\begin{theorem}[\cite{Property_Testing_and_its_Connection_to_Learning_and_Approximation}]
\label{general.thm.ggr}
Given inputs $G = (V, w)$ with $|V| = n$ and $w: V \times V \rightarrow [0,1]$ describing the graph and $\Phi$ describing bounds on the wanted partition, $\epsilon_{err}$, the algorithm $PT(G, \Phi, \epsilon_{err})$  has expected running time\footnote{We remark that the original algorithm contains a probability of error $\delta$, that appears in the running time. We disregard this error and bound the expected running time of the algorithm.}  of
\[
\exp\big( \log(\frac{1}{\epsilon_{err}}) \cdot (\frac{O(1)}{\epsilon_{err}})^{k+1}\big) + O(\frac{\log{\frac{k}{\epsilon_{err}}}}{\epsilon_{err}^2}) \cdot n.
\]
Furthermore, if $G \in \mathcal{GP}_{\Phi}$ as in Definition \ref{general.def.ggr_inequalities} then the algorithm outputs a partition satisfying
\begin{itemize}
    \item $\forall j: \,\, \lambda_j^{LB} - \epsilon_{err} \leq \frac{|V_j|}{n} \leq \lambda_j^{UB} + \epsilon_{err}$,
    \item $\forall j, j': \,\, \mu_{j,j'}^{LB} - \epsilon_{err} \leq \frac{W_{V_j, V_{j'}}}{n^2} \leq \mu_{j,j'}^{UB} + \epsilon_{err}$.
\end{itemize}
\end{theorem}
 
\noindent We are now ready to prove Proposition \ref{la.prop.2}.

\begin{proposition}
\label{la.prop.2}
If $ALG$ terminates in case (a) then $\frac{ALG_{d-w}(G_k)}{OPT(G_k)} = \frac{ALG(G_k)}{OPT(G_k)} \geq 1 - 20 \epsilon.$
\end{proposition}

\begin{proof}
Let $P = \{P_i\} $ denote the partition returned by $PT(G_k, \Phi, \epsilon_{err})$ and recall that its number of sets is $k = \frac{1}{\epsilon}$ and that $\epsilon_{err} = \epsilon^9$. We first observe that by Theorem \ref{general.thm.ggr} we are guaranteed that the error in $|P_i|$ compared to $|P^*_i| = \epsilon n$ is at most $|P_i| \geq \epsilon n - \epsilon_{err} n$ (due to the fact that in $\Phi$ we requested sets of size exactly $\epsilon n$). Therefore
\begin{equation}
\label{la.prop_2.1}
\begin{gathered}
ALG_{d-w} \geq 
\sum_{\substack{1\leq i \leq k-1\\ 1 \leq j \leq k-i}} W_{P_i, P_{i+j}} (|P_{i+1}| + \cdots + |P_{i+j-1}|) \geq
\sum_{\substack{1\leq i \leq k-1\\ 1 \leq j \leq k-i}} (j-1) (\epsilon n - \epsilon_{err} n) W_{P_i, P_{i+j}} ,
\end{gathered}
\end{equation}
where $W_{P_i, P_{i+j}}$ denotes the weight crossing between $P_i$ and $P_{i+j}$. For ease of presentation we will remove the subscript in the summation henceforth.

Consider the difference between the cut size of $W_{P_i, P_{i+j}}$ and $W_{P^*_i, P^*_{i+j}}$. Their difference originates from two errors: (1) the error that incurred by the PT algorithm (see Theorem \ref{general.thm.ggr}) and (2) the error $ALG_{d-w}$ incurred in order to guess the partition of $OPT(G_k)$ (see Algorithm \ref{alg.la_alg_dense}). Therefore, 
\[
W_{P_i, P_{i+j}} \geq W_{P^*_i, P^*_{i+j}} - \epsilon_{err} n^2 D_V - \epsilon^9 n^2 D_V = W_{P^*_i, P^*_{i+j}} - 2\epsilon^9 n^2 D_V
\]
where the last equality is since $\epsilon_{err} = \epsilon^9$. Combining this with inequality \ref{la.prop_2.1} yields
\begin{equation}
\label{la.prop_2.3}
\begin{gathered}
ALG_{d-w} \geq 
(\epsilon n - \epsilon_{err} n) \cdot \sum  (j-1)W_{P^*_i, P^*_{i+j}} - 
(\epsilon n - \epsilon_{err} n) \cdot 2(\epsilon^9 n^2) D_V \sum(j-1) \geq \\ 
(\epsilon n - \epsilon_{err} n) \cdot \sum  (j-1)W_{P^*_i, P^*_{i+j}} - 2n^3 \epsilon^7 D_V,
\end{gathered}
\end{equation}
where the last inequality follows since $\epsilon_{err} > 0$ and $\sum(j-1) = \sum_{i = 1}^k \sum_{j = i+1}^k (j-1) \leq k^3 = \epsilon^{-3}$.

Due to the fact that we are in case (a) we have that $\frac{W}{n^2 D_V} = \rho \geq \epsilon^6$. By Fact \ref{fact.avg_link} we have that $OPT \geq \frac13 nW$ and therefore $2n^3 \epsilon^7 D_V$ can be bounded by $2n^3 \epsilon^7 D_V \leq 2 \epsilon n W  \leq 6 \epsilon OPT$. Thus we get $2n^3 \epsilon^7 D_V \leq 6 \epsilon OPT(G_k)$. Combining this with inequality \ref{la.prop_2.3} yields
\begin{equation}
\label{la.prop_2.6}
ALG_{d-w} \geq 
(\epsilon n - \epsilon_{err} n) \cdot \sum (j-1) W_{P^*_i, P^*_{i+j}}  - 6 \epsilon OPT(G_k).
\end{equation}

On the other hand, recall that $P^*$ denotes the balanced partition where all sets are of size $\epsilon n$. Therefore, by Lemma \ref{la.lemma.1.5} we therefore get
\begin{equation}
\label{la.prop_2.2}
\begin{gathered}
OPT(G_k) \leq 
(1 + 13 \epsilon)\sum  W_{P^*_i, P^*_{i+j}} (|P^*_{i+1}| + \cdots + |P^*_{i+j-1}|) = \\
(1 + 13 \epsilon)\sum  (j-1) (\epsilon n) W_{P^*_i, P^*_{i+j}}  = 
\epsilon n(1 + 13 \epsilon)  \cdot \sum (j-1) W_{P^*_i, P^*_{i+j}} .
\end{gathered}
\end{equation}

% To complete the proof we will lower bound the RHS of inequality \ref{la.prop_2.3} using inequality \ref{la.prop_2.2}. First, we observe that $\sum(j-1) = \sum_{i = 1}^k \sum_{j = i+1}^k (j-1) \leq k^3$ and thus 
% \[
% (\epsilon n - \epsilon_{err} n) \cdot 2(\epsilon^9 n^2) \sum(j-1) \leq
% (\epsilon n - \epsilon_{err} n) \cdot 2(\epsilon^9 n^2) \cdot \frac{1}{\epsilon^3} \leq
% 2n^3 \epsilon^7,
% \]
% since $k = \frac{1}{\epsilon}$.

\noindent Combining inequalities \ref{la.prop_2.6} and \ref{la.prop_2.2} yields
\begin{equation*}
\begin{gathered}
ALG_{d-w} \geq
\frac{\epsilon n - \epsilon_{err} n}{\epsilon n(1 + 13 \epsilon)} OPT(G_k) - 6\epsilon OPT(G_k) = \\
\frac{1 - \epsilon^8}{1 + 13 \epsilon} OPT(G_k) - 6\epsilon OPT(G_k) \geq
(1 - 20 \epsilon)OPT(G_k),
\end{gathered}
\end{equation*}
thereby concluding the proof.
\end{proof}

\subsubsection{Analyzing the Approximation Ratio of Case (b) of $ALG$}

Using our structural lemmas we will analyze the approximation ratio of $ALG$ applied to $G_k$ under the assumption that the algorithm terminated in case (b) (i.e., that $\rho < \epsilon^6$ and $W_{B \cup C} \leq \epsilon W_{G_k}$). The full proof is deferred to the Appendix.

\begin{proposition}
\label{la.prop.3}
If $ALG$ terminates in case (b) then $\frac{ALG(G_k)}{OPT(G_k)} \geq 1 - 33 \epsilon.$
\end{proposition}

\begin{proof} [Sketch]
    The proof follows the following path. Due to the fact that most of the instance's density is centered at the metric's core $C$, the majority of $OPT(G_k)$'s objective is derived from weights incident to $C$. Since we are case (b), the weight of $W_{B \cup C}$ is negligible and therefore we will show that in fact $OPT(G_k)$'s objective is defined by $\opt{G_k}{A,C}$. Thereafter, we show that in fact the best strategy to optimize for weights in $W_{A,C}$ is to place $A$ at one extreme of the line and $C$ at the other - which, fortunately, is what $ALG(G_k)$ (approximately) does - thereby approximating $OPT(G_k)$.
\end{proof}

\subsubsection{Setting the Values $\alpha_i$, $\beta_i$ and $\gamma_i$}

Due to lack of space, the following proofs are deferred to the Appendix.

\begin{proposition}
\label{la.prop.4}
For $A_i$ and $C_i$ as defined by our algorithm applied to $G_i$ and for $\alpha_i = \frac12 n W_{A, C} (1 - \frac{5 \sqrt{\rho}}{\epsilon^2})$, we have $ALG(G_i) \geq \alpha_i + ALG(G_{i+1})$.
% \[
% ALG(G_i) \geq
% \alpha_i + ALG(G_{i+1}).
% \]
\end{proposition}

\begin{proposition}
\label{la.prop.8}
Let $G_i = (V_i, w_i)$ and $G_{i+1} = (V_{i+1}, w_{i+1})$ denote the instances defined by the $i$ and $i+1$ recursion steps. Furthermore let $\beta_i = \frac12 n_{V_i} W_{A_i, C_i} (1 + \frac{13 \sqrt{\rho}}{\epsilon^2})$ and $\gamma_i =  1 + 4 \sqrt{\rho_i}$. Therefore,
$
OPT(G_i) \leq 
\beta_i + \gamma_i OPT(G_{i+1}).
$
\end{proposition}

\noindent Thus, we have managed to set the values of $\alpha_i$, $\beta_i$ and $\gamma_i$ as follows.

\begin{definition}
\label{la.def.3}
We define the values $\alpha_i$, $\beta_i$ and $\gamma_i$ as follows

\begin{equation}
\label{la.eq.1}
\begin{gathered}
    \alpha_i = \frac12 n W_{A_i, C_i}(\conslaalphai); \quad
    \beta_i = \frac12 n_{V_i} W_{A_i, C_i} (1 + \frac{13 \sqrt{\rho_i}}{\epsilon^2}); \quad
    \gamma_i  =  1 + 4 \sqrt{\rho_i}.
\end{gathered}    
\end{equation}
\end{definition}

\subsubsection{Putting it all Together}

Now that we have analyzed the terminal cases of the algorithm (cases (a) and (b)) and that we have set the values of $\alpha_i$, $\beta_i$ and $\gamma_i$ we will to combine these results to prove $ALG$'s approximation ratio (as in Observation \ref{general.ob.0}). In order to so we must therefore bound the values $\min_i \{\frac{\alpha_i}{\beta_i \Pi_{j=0}^{i-1} \gamma_j}\}$ and $\frac{ALG(G_k)}{(\Pi_{i=0}^{k-1} \gamma_i)OPT(G_k)}$. However, before doing so we will first show that $\Pi_{j=0}^{i-1} \gamma_j$ converges. Recall that $\gamma_i = 1 + 4 \sqrt{\rho_i}$. The following lemma shows that the instances' densities ($\rho_i$) increase at a fast enough rate (exponentially) in order for $\Pi_{j=0}^{i-1} \gamma_j$ to converge. 

\begin{lemma}
\label{la.lemma.9}
For all $i = 1, \ldots, k-1$ we are guaranteed that $\rho_{i+1} \geq 4 \rho_i$.
\end{lemma}

\begin{proof}
Let $V$ denote the set of nodes of $G_i$. Recall the notations $A$, $B$ and $C$ defined by our algorithm applied to $V$ (in particular, the set of nodes of $G_{i+1}$ is exactly $B \cup C$). Therefore, if we denote by $D_{B-C}$ the largest distance between any point in $B$ and its closest point in $C$, then
$
D_{B \cup C} \leq 2D_{B-C} + D_C \leq 2\epsilon^2 D_V + 4D_V \sqrt{\rho_i},
$
where the first inequality follow from the triangle inequality and the second follows due to the fact that $B$ is defined as the set of all points of distance at most $\epsilon^2$ from $C$. Therefore,
\begin{equation}
\label{la.equation.lemma_9.1}
\begin{gathered}
\rho_{i+1} = 
\frac{W_{B \cup C}}{n^2_{B \cup C} \cdot D_{B \cup C}} \geq 
\frac{W_{V}}{n^2_{V} \cdot D_{V}} (\frac{\epsilon}{2\epsilon^2 + 4 \sqrt{\rho_i}}) = 
\rho_i (\frac{\epsilon}{2\epsilon^2 + 4 \sqrt{\rho_i}}),    
\end{gathered}
\end{equation}
where the equalities follows by the definition of $\rho_i$ and the inequality follows due to the fact that $W_{B \cup C} \geq \epsilon W_V$ (which follows due to the fact that we are in case (c)), $n_{B \cup C} \leq n_V$ and $D_{B \cup C} \leq (2\epsilon^2  + 4\sqrt{\rho_i})D_V$ (as stated above). Since we are in case (c), we are guaranteed that $\rho_i \leq \epsilon^6$ and therefore
\begin{equation}
\label{la.equation.lemma_9.2}
\begin{gathered}
\frac{\epsilon}{2\epsilon^2 + 4 \sqrt{\rho_i}} \geq
\frac{\epsilon}{2\epsilon^2 + 4 \epsilon^3} \geq
\frac{1}{3 \epsilon},
\end{gathered}
\end{equation}
since $\epsilon \leq 10^{-2}$. Combining inequalities \ref{la.equation.lemma_9.1} and \ref{la.equation.lemma_9.2}, and since $\epsilon < 10^{-2}$ yields
$
\rho_{i+1}  \geq
\rho_i (\frac{\epsilon}{2\epsilon^2 + 4 \sqrt{\rho_i}}) \geq
\frac{\rho_i}{3 \epsilon} \geq
4 \rho_i,
$
thereby concluding the proof.
\end{proof}

\noindent We are now ready to show that $\Pi_{j=0}^{i-1} \gamma_j$ converges. 

\begin{lemma}
\label{la.lemma.10}
For $\gamma_i = 1 + 4 \sqrt{\rho_i}$ we have $\Pi_{j=0}^{i-1} \gamma_j \leq 1 + 5 \sqrt{\rho_i}$.
\end{lemma}

\begin{proof}
Observe that 
$
\Pi_{j=0}^{i-1} (1 + 4\sqrt{\rho_j}) \leq
e^{4 \cdot \sum_j \sqrt{\rho_j}} \leq
e^{4 \sqrt{\rho_i}} \leq
1 + 5 \sqrt{\rho_i},
$
where the first inequality follows from Observation \ref{general.ob.0}, the second follows since $\sqrt{\rho_j}$ are exponentially increasing (Lemma \ref{hc.lemma.7}) and the third inequality follows again by Observation \ref{general.ob.0} combined with the fact that $\rho < \epsilon^2$ and $\epsilon < 10^{-2}$.
\end{proof}

\noindent Next we leverage the former lemma to bound $\min_i \{\frac{\alpha_i}{\beta_i \Pi_{j=0}^{i-1} \gamma_j}\}$ and $\frac{ALG(G_k)}{(\Pi_{i=0}^{k-1} \gamma_i)OPT(G_k)}$.

\begin{proposition}
\label{la.prop.11}
For $\alpha_i$, $\beta_i$ and $\gamma_i$ as in Definition \ref{la.def.3}, we have $\min_i \{\frac{\alpha_i}{\beta_i \Pi_{j=0}^{i-1} \gamma_j}\} \geq 
1 - 23\epsilon$.
\end{proposition}

% \begin{proposition}
% \label{la.prop.12}
% For $\alpha_i$, $\beta_i$ and $\gamma_i$ as in Definition \ref{la.def.3}, we have
% \[
% \min_i \{\frac{\alpha_i}{\beta_i \Pi_{j=0}^{i-1} \gamma_j}\} \geq 
% 1 - 100\epsilon.
% \]
% \end{proposition}

\begin{proof}
We first bound $\frac{\alpha_i}{\beta_i}$. By the definitions of $\alpha_i$ and $\beta_i$ we have
\begin{equation}
\label{la.eq.lemma_10.2}
\begin{gathered}
\frac{\alpha_i}{\beta_i} = 
\frac{1 - \frac{5 \sqrt{\rho_i}}{\epsilon^2}}{1 + \frac{13 \sqrt{\rho_i}}{\epsilon^2}} \geq 
(1 - \frac{5 \sqrt{\rho_i}}{\epsilon^2})(1 - \frac{13 \sqrt{\rho_i}}{\epsilon^2}) \geq
1 - \frac{18 \sqrt{\rho_i}}{\epsilon^2},
\end{gathered}
\end{equation}
where the first inequality follows from the definitions of $\alpha_i$ and $\beta_i$ and the rest of the inequalities follow since $\epsilon < 10^2$ and $\rho < \epsilon^6$.

By Lemma \ref{la.lemma.10} we are guaranteed that $\Pi_{j=0}^{i-1} \gamma_j \leq 1 + 5 \sqrt{\rho_i}$. Combining this with inequality \ref{la.eq.lemma_10.2} yields
\[
\frac{\alpha_i}{\beta_i \Pi_{j=0}^{i-1} \gamma_j} \geq 
\frac{1 - \frac{18 \sqrt{\rho_i}}{\epsilon^2}}{1 + 5 \sqrt{\rho_i}} \geq
(1 - \frac{18}{\epsilon^2} \sqrt{\rho_i})(1 - 5 \sqrt{\rho_i}) \geq
1 - \frac{23}{\epsilon^2} \sqrt{\rho_i},
\]
and since $\rho_i$ only increases and $\rho_{k-1} \leq \epsilon^6$ we have
$\min_i \{\frac{\alpha_i}{\beta_i \Pi_{j=0}^{i-1} \gamma_j}\} \geq 1 - \frac{23}{\epsilon^2} \sqrt{\rho_{k-1}} \geq 1 - 23 \epsilon$,
thereby concluding the proof.
\end{proof}

\begin{proposition}
\label{la.prop.12}
For $\gamma_i = 1 + 4 \sqrt{\rho_i}$ we have $\frac{ALG(G_k)}{(\Pi_{i=0}^{k-1} \gamma_i)OPT(G_k)} \geq 
1 - 34 \epsilon$.
\end{proposition}

\begin{proof}
By Propositions \ref{la.prop.2} and \ref{la.prop.3} we are guaranteed that $\frac{ALG(G_k)}{OPT(G_k)} \geq 1 - 33 \epsilon$. On the other hand by by Lemma \ref{la.lemma.10} we are guaranteed that $\Pi_{i=0}^{k-2} \gamma_i \leq 1 + 5 \sqrt{\rho_{k-1}}$. Therefore, if $k = 1$ then 
$\frac{ALG(G_k)}{(\Pi_{i=0}^{k-1} \gamma_i)OPT(G_k)} = 
\frac{ALG(G_k)}{OPT(G_k)} \geq
1 - 33 \epsilon.$ Otherwise, we have
\[
\frac{ALG(G_k)}{(\Pi_{i=0}^{k-1} \gamma_i)OPT(G_k)} \geq
\frac{1 - 33 \epsilon}{(1 + 4 \sqrt{\rho_{k-1}})(1 + 5 \sqrt{\rho_{k-1}})} \geq
\frac{1 - 33 \epsilon}{(1 + 4 \epsilon^3)(1 + 5 \epsilon^3 )} \geq
1 - 34 \epsilon,
\]
where the second inequality follows since $\rho_{k - 1} < \epsilon^6$ (since we recursed to step $k$) and the subsequent inequalities follow since $\epsilon < 10^{-3}$ - thereby concluding the proof.
\end{proof}

\noindent Finally, we combine Propositions \ref{la.prop.11} and \ref{la.prop.12} to bound $ALG$'s approximation ratio.

\begin{theorem}
\label{la.thm.approx_ratio}
For any metric $G$, $\frac{ALG(G)}{OPT(G)} \geq 1 - 34 \epsilon$.
\end{theorem}

\subsection{Analyzing the Running Time of $ALG$}
\label{la.sec.analyzing_runtime}

Consider the definition of $ALG$. We observe that in each recursion step, the algorithm finds the layer to peel off, $A$, and then recurses. Therefore the running time is defined by the sum of these recursion steps, plus the terminating cases (i.e., either case (a) or case (b)). Recall that case (a) applies $ALG_{d-w}$ on the instance, while case (b) arranges the instance arbitrarily. Therefore, a bound on cases (a) and (b) is simply a bound on the running time of $ALG_{d-w}$ which is given by Lemma \ref{la.lemma.11} (whose proof appears in the Appendix).

\begin{lemma}
\label{la.lemma.11}
Given an instance $G$, the running time of $ALG_{d-w}(G)$ is at most $(\frac{1}{\epsilon^7})^{\frac{1}{\epsilon^2}} \cdot O(n^2)$.
\end{lemma}

% \noindent Next we observe the following bound on the number of recursion steps performed by Algorithm \ref{alg.la_alg}.

\begin{remark}
\label{la.remark.loglog}
A bi-product of Lemma \ref{la.lemma.9} is that the number of recursion steps is bounded by $O(\log n)$. The proof follows similarly to the proof of Lemma \ref{hc.lemma.loglog} substituting the inequality $\rho_{i+1} \geq 4 \epsilon \sqrt{\rho_i}$ with $\rho_{i+1} \geq 4 \sqrt{\rho_i}$ (which holds due to Lemma \ref{la.lemma.9}).
\end{remark}

\noindent We are now ready to analyze the running time of $ALG$. (The proof is deferred to the Appendix.)

\begin{theorem}
\label{la.thm.running_time}
The algorithm $ALG$ is an EPRAS (with running time $O(n^2  \log n)$ plus the running time of $ALG_{d-w}$).
\end{theorem}

\begin{remark}
We remark that one may improve the running time by replacing $ALG_{d-w}$ with any faster algorithm while slightly degrading the quality of the approximation.
\end{remark}

\section{The Hierarchical Clustering Objective}
\label{hc.sec}

The section is outlined as follows. We begin by presenting our algorithms (first the algorithm to handle case (a) and subsequently the general algorithm). Thereafter we will bound the algorithm's approximation guarantee (by following the bounding scheme of Observation \ref{general.ob.1}). Finally, we will analyze the algorithm's running time.

\vspace*{-\baselineskip}

\subsection{Defining the Algorithms}
\label{hc.sec.defining_algorithms}

As in the linear arragement setting, we will begin by applying our general algorithm to the linear arrangement problem (which we will denote simply as $ALG$). The algorithm uses, as a subroutine, an algorithm to handle case (a). We denote this subroutine as $ALG_{d-w}$ and define it following the definition of $ALG$.

\vspace*{-\baselineskip}

\subsubsection{Defining $ALG$}

Here we apply our general algorithm (Algorithm \ref{alg.general_alg}) to the hierarchical clustering setting. In order to do so, roughly speaking, we define the layer to peel off $A$ as all points outside of the metric's core.\\

\vspace*{-\baselineskip}

\begin{algorithm}[H]
\caption{Hierarchical Clustering Algorithm ($ALG$)}
\label{alg.hc_alg}
% \If(\tcp*[f]{case (a)}){$\rho \geq (\frac{\epsilon}{80 \sqrt{2}})^2$}{
%     Solve it using $ALG_{d-w}$.
% }
\lIf(\tcp*[f]{case (a)}){$\rho \geq \epsilon^2$}{
    Solve the instance using $ALG_{d-w}$.
}
\Else{
    Let $C$ denote the metric's core (as defined by Lemma \ref{general.lemma.1}).
    
    Let $A = V \setminus C$ denote the rest of the points.

    Arrange $A$ as a (arbitrary) ladder and denote the tree by $T_A$.
    
    \uIf(\tcp*[f]{case (b)}){$W_{C} < 16 \epsilon\cdot W_V$}{
        Arrange $C$ arbitrarily and denote the resulting tree by $T_{C}$. 
        
        Attach $T_{C}$'s root as a child of the bottom most internal node of $T_A$ and return.

        % to the bottom of $T_A$ and return.
    }
    \Else(\tcp*[f]{case (c)}){
        Continue recursively on $C$ and denote the resulting tree by $T_{C}$. 
        
        Attach $T_{C}$'s root as a child of the bottom most internal node of $T_A$ and return.
    }
}
\end{algorithm}

\begin{remark}
Note that Algorithm \ref{alg.hc_alg} conforms to the general Algorithm \ref{alg.general_alg} since $C = V \setminus A$.
\end{remark}

\vspace*{-\baselineskip}

\subsubsection{Defining $ALG_{d-w}$}
We will use the algorithm of Vainstein et al. \cite{Hierarchical_Clustering_via_Sketches_and_Hierarchical_Correlation_Clustering} as $ALG_{d-w}$. As part of their algorithm they make use of the general graph partitioning algorithm of Goldreich et al. \cite{Property_Testing_and_its_Connection_to_Learning_and_Approximation} which is denoted by $PT(\cdot)$. Since we will use $PT(\cdot)$ to devise our own algorithm for the LA objective we refer the reader to Definition \ref{general.def.ggr_inequalities} and Theorem \ref{general.thm.ggr} for a more in-depth explanation of the $PT(\cdot)$ algorithm. We restate $ALG_{d-w}$ in Algorithm \ref{alg.hc_alg_dense}  as defined in Vainstein et al. \cite{Hierarchical_Clustering_via_Sketches_and_Hierarchical_Correlation_Clustering}.

\begin{algorithm}[H]
\caption{HC Algorithm for Sufficiently Densely Weighted Instances ($ALG_{d-w}$)}
\label{alg.hc_alg_dense}
Enumerate over all trees $T$ with $k = \frac{1}{\epsilon}$ internal nodes.

\For{each such $T$}{
    \For{$\{\lambda_i \}_{i \leq k} \subset \{i \epsilon^2 n : i \in \mathbb{N} \land i \leq \frac{3}{\epsilon}\}$}{
        \For{$\{\mu_{j,j'}\}_{j\leq k, j' \leq k, j\neq j'} \subset \{i \epsilon^3 n^2 D_V: i \in \mathbb{N} \land i \leq \frac{9}{\epsilon}\}$}{
            Let $\Phi = \{\lambda_i, \lambda_i \}_{i = 1}^k \cup \{\mu_{j,j'}, \mu_{j,j'} \}_{j,j' = 1}^k$.
        
            Run $PT(G, \Phi, \epsilon_{err} = \epsilon^3)$. Let $P$ denote the output partition (if succeeded).
            
            Compute the HC objective value based on $T$ and $P$.
        }
    }
}

Return the partition $P$ and tree $T$ with maximal HC objective value.
\end{algorithm}

\subsection{Analyzing the Approximation Ratio of $ALG$}
\label{hc.sec.analyzing_approx}

Now that we have defined $ALG$ we are ready to analyze its approximation ratio. Recall that by Observation \ref{general.ob.1} it is enough to analyze the approximation ratio of cases (a), (b) and the total approximation loss generated by the recursion steps (i.e., by finding $\alpha_i$, $\beta_i$ and $\gamma_i$).

\subsubsection{Analyzing the Approximation Ratio of Case (a) of $ALG$}

In order to analyse the approximation ratio of $ALG_{d-w}$ in our setting we must first recall the definition of instances with not-all-small-weights (as defined by Vainstein et al. \cite{Hierarchical_Clustering_via_Sketches_and_Hierarchical_Correlation_Clustering}).

\begin{definition}
A metric $G$ is said to have not all small weights if there exist constants (with respect to $n_V$) $c_0, c_1 < 1$ such that the fraction of weights smaller than $c_0 \cdot D_V$ is at most $1-c_1$.
\end{definition}

\noindent The following theorem was presented in Vainstein et al. \cite{Hierarchical_Clustering_via_Sketches_and_Hierarchical_Correlation_Clustering}.

\begin{theorem}
\label{hc.theorem.1}
For any constant $\xi>0$ and any metric $G = (V,w)$ with not all small weights (with constants $c_0$ and $c_1$) we are guaranteed that
$\frac{ALG_{d-w}(G)}{OPT(G)} \geq 1 - O(\frac{\xi}{c_0 \cdot c_1})$ and that $ALG_{d-w}$'s expected running time is at most $f(\frac{1}{\xi}) \cdot n^2$.
\end{theorem}

\noindent Applying the above theorem with $\xi = \epsilon^5$ to our metric instance $G_k$ yields Proposition \ref{hc.proposition.2} (whose proof is deferred to the Appendix).

\begin{proposition}
\label{hc.proposition.2}
If $ALG$ terminates in case (a) then $\frac{ALG_{d-w}(G_k)}{OPT(G_k)} = \frac{ALG(G_k)}{OPT(G_k)} \geq 1 - \epsilon.$
\end{proposition}

\subsubsection{Analyzing the Approximation Ratio of Case (b) of $ALG$}

\begin{proposition}
\label{hc.prop.3}
If $ALG$ terminates in case (b) then $\frac{ALG(G_k)}{OPT(G_k)} \geq 1 - 17\epsilon.$
\end{proposition}

\begin{proof}
    The proof is deferred to the Appendix.
\end{proof}

\subsubsection{Setting the Values $\alpha_i$, $\beta_i$ and $\gamma_i$}

Due to lack of space, we defer the following proofs to the Appendix.

\begin{lemma}
\label{hc.lemma.4}
For $A_i$ and $C_i$ as defined by our algorithm applied to $G_i$ and for $\alpha_i = n_{V_i} (W_{A_i} + W_{A_i,C_i})(1 - \sqrt{\rho_i})$ we have $ALG(G_i) \geq \alpha_i + ALG(G_{i+1})$.
\end{lemma}

\begin{lemma}
\label{hc.lemma.6}
Let $G_i = (V_i, w_i)$ and $G_{i+1} = (V_{i+1}, w_{i+1})$ denote the instances defined by the $i$ and $i+1$ recursion steps. Furthermore, let $\beta_i = n_{V_i}(W_{A_i} + W_{A_i, C_i})$ and $\gamma_i = 1 + 2 \sqrt{\rho_i}$. Therefore,
$OPT(G_i) \leq \beta_i + \gamma_i OPT(G_{i+1}).$

\end{lemma}

\noindent Thus, we combine these values in Definition \ref{hc.def.3}.

\begin{definition}
\label{hc.def.3}
We define the values $\alpha_i$, $\beta_i$ and $\gamma_i$ as follows
\begin{equation*}
\begin{gathered}
    \alpha_i = n_{V_i} (W_{A_i} + W_{A_i,C_i})(1 - \sqrt{\rho_i}); \quad
    \beta_i =  n_{V_i}(W_{A_i} + W_{A_i, C_i}); \quad
    \gamma_i  = 1 + 2 \sqrt{\rho_i}.
\end{gathered}    
\end{equation*}
\end{definition}

\subsubsection{Putting it all Together}

Now that we have analyzed the terminal cases of the algorithm (cases (a) and (b)) and that we have set the values of $\alpha_i$, $\beta_i$ and $\gamma_i$ we will combine these results to prove $ALG$'s approximation ratio (as in Observation \ref{general.ob.1}). Due to lack of space we defer the proofs of this section to the Appendix.

\begin{proposition}
\label{hc.prop.11}
For $\alpha_i$, $\beta_i$ and $\gamma_i$ as in Definition \ref{hc.def.3}, we have $\min_i \{\frac{\alpha_i}{\beta_i \Pi_{j=0}^{i-1} \gamma_j}\} \geq 
1 - 4 \epsilon$.
\end{proposition}

\begin{proposition}
\label{hc.prop.12}
For $\gamma_i = 1 + 2 \sqrt{\rho_i}$ we have $\frac{ALG(G_k)}{(\Pi_{i=0}^{k-1} \gamma_i)OPT(G_k)} \geq 
1 - 23 \epsilon$.
\end{proposition}

\begin{theorem}
\label{hc.thm.last}
For any metric $G$, $\frac{ALG(G)}{OPT(G)} \geq 1 - 23 \epsilon$.
\end{theorem}

\subsection{Analyzing the Running Time of $ALG$}
\label{hc.sec.analyzing_runtime}

Consider the definition of $ALG$. In each recursion step, the algorithm finds the layer to peel off and then recurses. Therefore the running time is defined by the sum of these recursion steps, plus the terminating cases (i.e., either case (a) or case (b)). Recall that case (a) applies $ALG_{d-w}$ on the instance, while case (b) arranges the instance arbitrarily. Therefore, a bound on cases (a) and (b) is simply a bound on the running time of $ALG_{d-w}$ which is given by Theorem \ref{hc.theorem.1} \cite{Hierarchical_Clustering_via_Sketches_and_Hierarchical_Correlation_Clustering}. In Lemma \ref{hc.lemma.loglog} we bound the number of recursion steps and subsequently prove Theorem \ref{hc.thm.running_time} (proofs appear in the Appendix).

\begin{theorem}
\label{hc.thm.running_time}
The algorithm $ALG$ is an EPRAS (with running time $O(n^2  \log \log n)$ plus the running time of $ALG_{d-w}$).
\end{theorem}

\begin{remark}
We remark that one may improve the running time by replacing $ALG_{d-w}$ with any faster algorithm while slightly degrading the quality of the approximation.
\end{remark}

\bibliographystyle{plainnat}
\bibliography{bib.bib}

\appendix

\section{Deferred Proofs of Section \ref{sec.multi_layer_peeling_framework}}

\begin{proof}[Proof of Lemma \ref{general.lemma.1}]
For every node $v$ let $W_v$ denote the sum of weights incident to $v$. Let $U_v$ denote the set of nodes that are within $2D_V\sqrt{\rho_V}$ distance of $v$. We prove that there exists $v$ with $|U_v| \geq n_V(1 - \sqrt{\rho_V})$, thereby concluding the proof. 

Assume towards contradiction that this is not the case. Then, for every node $v$ we have $$W_v > (\sqrt{\rho_V} n_V)(2D_V\sqrt{\rho_V}) = 2n_VD_V\rho_V$$ (since there are at least $n_V \sqrt{\rho_V}$ nodes of distances $\geq 2D_V\sqrt{\rho_V}$ from $v$). Summing over all $v$ yields $2W_V = \sum_v W_v > 2n_V^2 D_V \rho_V = 2W_V$ which is a contradiction. 
\end{proof}

%%%%%%%%%%%%%%%%%%%%%%%%%%%%%%%%%%%%%%%%%%%%%%%%%%%%%%%%%%%%%%%%%%%%%%%%%%%%%%%%%%%%%%%%%%%%%%%%%%%%%%%%%%%%%%%%%%%
%%%%%%%%%%%%%%%%%%%%%%%%%%%%%%%%%%%%             LA SECTION START          %%%%%%%%%%%%%%%%%%%%%%%%%%%%%%%%%%%%%%%%
%%%%%%%%%%%%%%%%%%%%%%%%%%%%%%%%%%%%%%%%%%%%%%%%%%%%%%%%%%%%%%%%%%%%%%%%%%%%%%%%%%%%%%%%%%%%%%%%%%%%%%%%%%%%%%%%%%%

\section{Deferred Proofs of Section \ref{la.sec}}

\begin{proof}[Proof of Lemma \ref{la.lemma.6}]
We first observe that for any $a \in A_i$ we have 
\begin{equation}
\label{la.eq.lemma_6_1}
\begin{gathered}
\sum_{c \in C_i} w_{a,c} y_{a,c} \geq 
(\min_{c \in C_i} \{ w_{a,c} \}) \sum_{c \in C_i} y_{a,c} \geq \\
(\min_{c \in C_i} \{ w_{a,c} \}) (1 + 2 + 3 + \cdots + n_{C_i}) \geq 
(\min_{c \in C_i} \{ w_{a,c} \}) \frac{n^2_{C_i}}{2},
\end{gathered}
\end{equation}
where the second inequality follows since $y$ places all the points in $A_i$ to the left of all the points in $C_i$.

By the triangle inequality for any point $c \in C_i$ we have that $\min_{c \in C_i} \{ w_{a,c} \} + D_{C_i} \geq w_{a,c}$ and therefore
\begin{equation}
\label{la.eq.lemma_6_2}
(\min_{c \in C_i} \{ w_{a,c} \}) \cdot n_{C_i} \geq
\sum_{c \in C_i} (w_{a,c} - D_{C_i}).
\end{equation}

\noindent Therefore, by summing over all $a \in A_i$
\begin{equation*}
\begin{gathered}
\sum_{c \in C_i,a \in A_i} w_{a,c} y_{a,c} =
\sum_{a \in A_i} \sum_{c \in C_i} w_{a,c} y_{a,c} \geq \\
\sum_{a \in A_i} (\min_{c \in C_i} \{ w_{a,c} \}) \frac{n^2_{C_i}}{2} \geq
\sum_{a \in A_i} \big(\sum_{c \in C_i} (w_{a,c} - D_{C_i}) \frac{n_{C_i}}{2} \big) = \\ 
\frac{n_{C_i}}{2}(W_{C_i,A_i} - n_{C_i} n_{A_i} D_{C_i}),
\end{gathered}
\end{equation*}
where the first inequality follows from inequality \ref{la.eq.lemma_6_1} and the second follows from inequality \ref{la.eq.lemma_6_2} - thereby concluding the proof.

\end{proof}

\begin{proof}[Proof of Corollary \ref{la.cor.0}]
We begin with Lemma \ref{la.lemma.6}
\begin{equation}
\label{la.eq.cor_0.0}
\begin{gathered}
\sum_{a \in A, c \in C} w_{a,c} y_{a,c} \geq 
\frac{n_C}{2}(W_{A,C} - n_A n_C D_C).
\end{gathered}
\end{equation}

Recall that by the definitions of $A$ and $C$ all weights between sets $A$ and $C$ are at least $\epsilon^2 D_V$ and therefore $W_{A,C} \geq n_A n_C D_V \epsilon^2$. Further recall that by Lemma \ref{general.lemma.1} we are guaranteed that $D_C \leq 4\sqrt{\rho} D_V$ and therefore 
\begin{equation}
\label{la.eq.cor_0.1}
\begin{gathered}
n_A n_C D_C \leq 
n_A n_C 4\sqrt{\rho}D_V \leq  
\frac{4 \sqrt{\rho}}{\epsilon^2} W_{A,C}. 
\end{gathered}
\end{equation}

\noindent Combining inequalities \ref{la.eq.cor_0.0} and \ref{la.eq.cor_0.1} with the fact that $n_C \geq (1 - \sqrt{\rho})n$ yields
\begin{equation*}
\begin{gathered}
\sum_{a \in A, c \in C} w_{a,c} y_{a,c} \geq
\frac{n_C}{2}(W_{A,C} - n_A n_C D_C) \geq \\
\frac{n_C}{2}W_{A,C}(1 - \frac{4 \sqrt{\rho}}{\epsilon^2}) \geq
\frac12 n W_{A,C} (1 - \sqrt{\rho})(1 - \frac{4 \sqrt{\rho}}{\epsilon^2}).
\end{gathered}
\end{equation*}

\noindent Finally, since $\epsilon < 1$ we have
\[
\sum_{a \in A, c \in C} w_{a,c} y_{a,c} \geq 
\frac12 n W_{A,C} (1 - \sqrt{\rho})(1 - \frac{4 \sqrt{\rho}}{\epsilon^2}) \geq
\frac12 n W_{A,C} (1 - \frac{5 \sqrt{\rho}}{\epsilon^2}),
\]
thereby concluding the proof.

\end{proof}

\begin{proof}[Proof of Lemma \ref{la.lemma.6_5}]
We first observe that
\begin{equation}
\label{la.eq.lemma_6_5.1}
\begin{gathered}
\sum_{c \in C_i} w_{p, c} y_{p,c} \leq 
(\max_{c \in C_i} \{w_{p, c}\}) \sum_{c \in C_i} y_{p,c} \leq \\
(\max_{c \in C_i} \{w_{p, c}\}) (\sum_{i=0}^{n_{C_i}-1} (n - n_{C_i} + i)) = 
(\max_{c \in C_i} \{w_{p, c}\}) (n \cdot n_{C_i} - \frac{n_{C_i}(n_{C_i}+1)}{2}) \leq \\
(\max_{c \in C_i} \{w_{p, c}\}) (n \cdot n_{C_i} - \frac{n_{C_i}^2}{2}),
\end{gathered}
\end{equation}
where the second inequality follows from the fact that to maximize $\sum_{c \in C_i} y_{p,c}$ (i.e., the inter-objective-value where all weights are equal to 1) one must place $p$ at one extreme of the line and $C_i$ at the other extreme. 
On the other hand, for every $c \in C_i$, by the triangle inequality we have $w_{p, c} + D_m \geq (\max_{c \in C_i} \{w_{p, c}\})$ and therefore,
\begin{equation}
\label{la.eq.lemma_6_5.2}
\begin{gathered}
n_{C_i} \cdot (\max_{c \in C_i} \{w_{p, c}\}) \leq
\sum_{c \in C_i} (w_{p, c} + D_{C_i}).
\end{gathered}
\end{equation}

\noindent Combining inequalities \ref{la.eq.lemma_6_5.1} and \ref{la.eq.lemma_6_5.2} yields
\begin{equation*}
\label{la.eq.lemma_6_5_3}
\begin{gathered}
\sum_{c \in C_i} w_{p, c} y_{p,c} \leq
(\max_{c \in C_i} \{w_{p, c}\}) (n \cdot  n_{C_i} - \frac{n_{C_i}^2}{2}) \leq \\ 
\sum_{c \in C_i}(w_{p, c} + D_{C_i}) (n - \frac{n_{C_i}}{2}) = 
(W_{p,C_i} + n_{C_i} D_{C_i})(n - \frac{n_{C_i}}{2}),
\end{gathered}
\end{equation*}
thereby concluding the proof.

\end{proof}

\begin{proof}[Proof of Lemma \ref{la.lemma.7}]
Follows by applying Lemma \ref{la.lemma.6_5} to all points $a \in {A_i}$.
\end{proof}

\begin{proof}[Proof of Corollary \ref{la.cor.1}]
We begin with Lemma \ref{la.lemma.7}
\begin{equation}
\label{la.eq.cor_1.0}
\begin{gathered}
\sum_{a \in A, c \in C} w_{a,c} y_{a,c} \leq 
(n - \frac{n_C}{2})(W_{A,C} + n_A n_C D_C).
\end{gathered}
\end{equation}

Recall that by the definitions of $A$ and $C$ all weights between sets $A$ and $C$ are at least $\epsilon^2 D_V$ and therefore $W_{A,C} \geq n_A n_C D_V \epsilon^2$. Further recall that by Lemma \ref{general.lemma.1} we are guaranteed that $D_C \leq 4\sqrt{\rho} D_V$ and therefore 
\begin{equation}
\label{la.eq.cor_1.1}
\begin{gathered}
n_A n_C D_C \leq 
n_A n_C 4\sqrt{\rho}D_V \leq  
\frac{4 \sqrt{\rho}}{\epsilon^2} W_{A,C}. 
\end{gathered}
\end{equation}

\noindent Combining inequalities \ref{la.eq.cor_1.0} and \ref{la.eq.cor_1.1} with the fact that $n_C \geq (1 - \sqrt{\rho})n$ yields
\begin{equation*}
\begin{gathered}
\sum_{a \in A, c \in C} w_{a,c} y_{a,c} \leq 
(n - \frac{n_C}{2})(W_{A,C} + n_A n_C D_C) \leq \\
(n - \frac{n_C}{2})W_{A,C}(1 + \frac{4\sqrt{\rho}}{\epsilon^2}) \leq
\frac12 n W_{A,C} (1 + \sqrt{\rho})(1 + \frac{4\sqrt{\rho}}{\epsilon^2}) \leq 
\frac12 n W_{A,C} (1 + \frac{9 \sqrt{\rho}}{\epsilon^2}).
\end{gathered}
\end{equation*}
where the last inequality follows since $\rho < 1$ and $\epsilon < 1$ - thereby concluding the proof.
\end{proof}

\begin{proof}[Proof of Lemma \ref{la.lemma.1.5}]
We first observe that $OPT(G_k)$ can be rewritten as
\begin{equation*}
\begin{gathered}
OPT(G_k) = 
\sum_{\substack{1\leq i \leq k-1\\ 1 \leq j \leq k-i}} \opt{G_k}{P_i^*, P_{i+j}^*} + \sum_{\substack{1\leq i \leq k}} \opt{G_k}{P_i^*}.
\end{gathered}
\end{equation*}
For ease of presentation we will remove the subscript in the summation henceforth. Due to the fact that $|P_i^*| \leq \epsilon n$ we have that $\sum_i \opt{G_k}{P_i^*} \leq \sum_i \epsilon n W_{P_i^*} \leq \epsilon n W$. Combining this with Fact \ref{fact.avg_link} guarantees that $\sum_i \opt{G_k}{P_i^*} \leq 3 \epsilon OPT(G_k)$. Therefore
\begin{equation*}
\begin{gathered}
OPT(G_k) = 
\sum \opt{G_k}{P_i^*, P_{i+j}^*} +  \sum \opt{G_k}{P_i^*} \leq \\
\sum \opt{G_k}{P_i^*, P_{i+j}^*} + 3 \epsilon OPT(G_k) \Rightarrow 
OPT(G_k) \leq \frac{1}{1 - 3 \epsilon }\sum \opt{G_k}{P_i^*, P_{i+j}^*}.
\end{gathered}
\end{equation*}

On the other hand every weight that crosses between $P_i^*$ and $P_{i+j}^*$ can contribute at most $(j+1)\epsilon n $ to the objective and therefore $\opt{G_k}{P_i^*, P_{i+j}^*} \leq W_{P^*_i, P^*_{i+j}} ((j+1)\epsilon n) = W_{P^*_i, P^*_{i+j}} (|P^*_{i}| + \cdots + |P^*_{i+j}|)$. Putting it all together gives us
\begin{equation}
\label{la.eq.lemma_1_5.1}
\begin{gathered}
OPT(G_k) \leq 
\frac{1}{1 - 3 \epsilon }\sum \opt{G_k}{P_i^*, P_{i+j}^*} \leq
 (1 + 4 \epsilon) \sum W_{P^*_i, P^*_{i+j}} (|P^*_{i}| + \cdots + |P^*_{i+j}|),
\end{gathered}
\end{equation}
where the last inequality follows since $\epsilon < \frac{1}{12}$. To conclude the proof we bound the value 
$(1 + 4 \epsilon) \sum W_{P^*_i, P^*_{i+j}} (|P^*_{i}|+ |P^*_{i+j}|)$. Recall that $|P^*_{i}|+ |P^*_{i+j}| = 2 \epsilon n$. Further note that $\sum W_{P^*_i, P^*_{i+j}} \leq W$ simply since every weight is counted at most once. Therefore
\begin{equation}
\label{la.eq.lemma_1_5.2}
\begin{gathered}
(1 + 4 \epsilon) \sum W_{P^*_i, P^*_{i+j}} (|P^*_{i}|+ |P^*_{i+j}|) =
(1 + 4 \epsilon) 2 \epsilon n \sum W_{P^*_i, P^*_{i+j}} \leq 
(1 + 4 \epsilon) 2 \epsilon n W.
\end{gathered}
\end{equation}

\noindent To conclude the proof, we use the fact that $OPT(G_k) \geq \frac{1}{3}nW$ (see Fact \ref{fact.avg_link}) and get
\begin{equation}
\label{la.eq.lemma_1_5.3}
\begin{gathered}
(1 + 4 \epsilon) 2 \epsilon n W \leq 
(1 + 4 \epsilon) 6 \epsilon OPT(G_k) \leq 
7 \epsilon OPT(G_k),
\end{gathered}
\end{equation}
for $\epsilon < 10^{-2}$. Combining the above inequalities \ref{la.eq.lemma_1_5.1}, \ref{la.eq.lemma_1_5.2} and \ref{la.eq.lemma_1_5.3} yields
\begin{equation*}
\begin{gathered}
OPT(G_k) \leq 
(1 + 4 \epsilon) \sum W_{P^*_i, P^*_{i+j}} (|P^*_{i+1}| + \cdots + |P^*_{i+j-1}|) +
7 \epsilon OPT(G_k) \Rightarrow \\
OPT(G_k) \leq 
\frac{1 + 4 \epsilon}{1 - 7 \epsilon} \sum W_{P^*_i, P^*_{i+j}} (|P^*_{i+1}| + \cdots + |P^*_{i+j-1}|)  \leq \\
(1 + 13\epsilon)\sum W_{P^*_i, P^*_{i+j}} (|P^*_{i+1}| + \cdots + |P^*_{i+j-1}|),
\end{gathered}
\end{equation*}
for $\epsilon < 10^{-2}$ - thereby concluding the proof.
\end{proof}

\begin{proof}[Proof of Proposition \ref{la.prop.3}]
Our proof will contain three steps - (1) we will show that $OPT(G_k) \leq (1 + 16 \epsilon)  \opt{G_k}{A,C}$ and (2) we will show that $\opt{G_k}{A,C} \leq (1 + \frac{16 \sqrt{\rho}}{\epsilon^2}) \alg{G_k}{A,C}$. In step (3) we combine these observations and prove the proposition.

\begin{enumerate}
    \item  $\mathbf{OPT(G_k) \leq ( 1 + 16\epsilon)  \opt{G_k}{A,C}}$: In order to show this we will first show that $W \leq 6 \epsilon W_{A,C}$ (and since the majority of the instance's weight is contained within $W_{A,C}$, $OPT(G_k)$ will generate most of its value from those weights). 
    
    Indeed, since we are in case (b) we have that $W_{B \cup C} \leq \epsilon W_V$ and therefore 
    \[
    W_{B \cup C} \leq \frac{\epsilon}{1 - \epsilon} (W_A + W_{A,B} + W_{A, C}) \leq 2 \epsilon (W_A + W_{A,B} + W_{A, C}),
    \]
    where the first inequality is since $W_V = W_{B \cup C} + W_A + W_{A,B} + W_{A, C}$ and the second inequality follows since $\epsilon < 10^{-3}$. On the other hand by Lemma \ref{la.lemma.7.5.5} we have that $W_A + W_{A,B} \leq 2\frac{\sqrt{\rho}}{\epsilon^2}W_{A, C}$. Therefore,
    \begin{equation}
    \label{la.eq.lemma_3.1}
    \begin{gathered}
    W_A + W_{A,B} + W_{B \cup C} \leq
    (1+ 2\epsilon)W_A +(1+ 2\epsilon)W_{A,B} + 2\epsilon W_{A, C}\leq \\
    (1+2\epsilon )(\frac{2\sqrt{\rho}}{\epsilon^2})W_{A, C} + 2\epsilon W_{A, C} \leq
    \Big((1+2\epsilon )(2 \epsilon) + 2\epsilon\Big)W_{A, C} \leq
    5 \epsilon W_{A,  C},
    \end{gathered}
    \end{equation}
    where the last inequality follows since we are in case (b) and therefore $\rho \leq \epsilon^6$ and since $\epsilon < \frac14$.

    \noindent Combining all the above yields
    \begin{equation*}
    \begin{gathered}
    OPT(G_k) = 
    \opt{G_k}{A} + \opt{G_k}{A, B} +  \opt{G_k}{B \cup C} + \opt{G_k}{A, C} \leq \\
    n(W_A + W_{A,B} + W_{B \cup C}) + \opt{G_k}{A, C} \leq
    5\epsilon n W_{A,C} + \opt{G_k}{A, C} \leq \\
    15\epsilon OPT(G_k) + \opt{G_k}{A, C},
    \end{gathered}
    \end{equation*}
    where the first inequality follows by simply rearranging $OPT(G_k)$'s terms and the second inequality follows simply since all $y_{i,j} \leq n$. The third inequality follows from inequality \ref{la.eq.lemma_3.1} and the last inequality follows from Fact \ref{fact.avg_link}. Rearranging the terms yields
    \begin{equation}
    \label{la.eq.lemma_3.2}
    \begin{gathered}
    OPT(G_k) \leq
    \frac{1}{1 - 15\epsilon} \opt{G_k}{A, C} \leq
    (1 + 16 \epsilon)\opt{G_k}{A, C},
    \end{gathered}
    \end{equation}
    where the last inequality follows since $\epsilon < 10^{-4}$.

    \item  $\mathbf{\opt{G_k}{A,C} \leq (1 + \frac{16 \sqrt{\rho}}{\epsilon^2} ) \alg{G_k}{A,C}}$: We do this by applying Corollary \ref{la.cor.0} to $ALG$'s arrangement which we will denote by $y^{ALG}$ and applying Corollary \ref{la.cor.1} to $OPT$'s arrangement (denoted by $y$). The two corollaries respectively yield
    \begin{equation*}
    \begin{gathered}
    \alg{G_k}{A,C} \geq 
    \frac12 n W_{A,C} (1 - \frac{5 \sqrt{\rho}}{\epsilon^2}), \\
    \opt{G_k}{A, C} \leq 
    \frac12 nW_{A,C}(1 + \frac{9 \sqrt{\rho}}{\epsilon^2}).
    \end{gathered}
    \end{equation*}
    Combining the two inequalities gives us
    \begin{equation}
    \label{la.eq.lemma_3.3}
    \begin{gathered}
    \opt{G_k}{A, C} \leq 
    (\frac{1 + \frac{9 \sqrt{\rho}}{\epsilon^2}}{1 - \frac{5 \sqrt{\rho}}{\epsilon^2}})\alg{G_k}{A,C} \leq
    (1 + \frac{16 \sqrt{\rho}}{\epsilon^2})\alg{G_k}{A,C},
    \end{gathered}
    \end{equation}
    where the last inequality follows since $\rho < \epsilon^6$ and $\epsilon < 10^{-3}$.
    
    \item $\mathbf{OPT(G_k) \leq  (1+33 \epsilon)ALG}$: We prove this by combining inequalities \ref{la.eq.lemma_3.1} and \ref{la.eq.lemma_3.3} - 
    \begin{equation*}
    \begin{gathered}
    OPT(G_k) \leq
    (1 + 16 \epsilon)\opt{G_k}{A, C} \leq \\
    (1 + 16 \epsilon)(1 + \frac{16 \sqrt{\rho}}{\epsilon^2})\alg{G_k}{A,C} \leq
    (1 + 16 \epsilon)(1 + \frac{16 \sqrt{\rho}}{\epsilon^2})ALG(G_k),
    \end{gathered}
    \end{equation*}
    and since $\rho < \epsilon^6$ and $\epsilon < 10^{-3}$ we get
    \[
    OPT(G_k) \leq
    (1+35\epsilon) ALG(G_k) \Rightarrow
    \frac{ALG(G_k)}{OPT(G_k)} \geq 1 - 33\epsilon,
    \]
    thereby concluding the proof. 
\end{enumerate}
\end{proof}

\begin{proof}[Proof of Proposition \ref{la.prop.4}]
Due to the fact that $ALG(G_i)$ recurses $B \cup C$ we have that $\alg{G_i}{B \cup C} = ALG(G_{i+1})$. Therefore,
\begin{equation*}
\begin{gathered}
ALG(G_i) = 
\alg{G_i}{A} + \alg{G_i}{A, B \cup C} + \alg{G_i}{B \cup C} = \\
\alg{G_i}{A} + \alg{G_i}{A, B \cup C} + ALG(G_{i+1}) \geq
\alg{G_i}{A, C} + ALG(G_{i+1}).
\end{gathered}
\end{equation*}
Applying Corollary \ref{la.cor.0} to $ALG(G_i)$'s arrangement results in
\begin{equation*}
\begin{gathered}
\alg{G_i}{A,C} \geq 
\frac12 n W_{A, C} (1 - \frac{5 \sqrt{\rho}}{\epsilon^2}).
\end{gathered}
\end{equation*}
Combining the two inequalities concludes the proof.
\end{proof}

\begin{proof}[Proof of Proposition \ref{la.prop.8}]
We first observe that 
\begin{equation}
\label{la.eq.prop.8.0}
\begin{gathered}
OPT(G_i) = 
\opt{G_i}{A_i} + \opt{G_i}{A_i,B_i} + \opt{G_i}{A_i,C_i} + \opt{G_i}{B_i \cup C_i}.
\end{gathered}
\end{equation}

Consider $\opt{G_i}{B_i \cup C_i} = \sum_{e \in B_i \cup C_i} w_e y_e$. The value $y_e$ is comprised of nodes from $n_{A_i}$ and $n_{B_i \cup C_i}$. Therefore $\opt{G_i}{B_i \cup C_i} \leq n_{A_i} W_{B_i \cup C_i} + \optsin{G_{i+1}}$, since $\optsin{G_{i+1}}$ solves the instance defined by $B_i \cup C_i$ optimally. By Fact \ref{fact.avg_link} we have that $\optsin{G_{i+1}} \geq \frac13 n_{B_i \cup C_i} W_{B_i \cup C_i}$. Additionally, by Lemma \ref{general.lemma.1}, we have $n_{B_i \cup C_i} \geq \frac{1 - \sqrt{\rho_i}}{\sqrt{\rho_i}}  n_{A_i}$. Combining the above yields
\begin{equation}
\label{la.eq.prop.8.1}
\begin{gathered}
\opt{G_i}{B_i \cup C_i} \leq 
n_{A_i} W_{B_i \cup C_i} + \optsin{G_{i+1}} \leq \\
\frac{\sqrt{\rho_i}}{1 - \sqrt{\rho_i}} n_{B_i \cup C_i} W_{B_i \cup C_i}+ \optsin{G_{i+1}} \leq 
(1 + \frac{3 \sqrt{\rho_i}}{1 - \sqrt{\rho_i}}) \optsin{G_{i+1}}.
\end{gathered}
\end{equation}

Next consider $\opt{G_i}{A_i} + \opt{G_i}{A_i,B_i}$. Observe that $\opt{G_i}{A_i} + \opt{G_i}{A_i,B_i} \leq n_{V_i}(W_{A_i} + W_{A_i, B_i})$ since every edge may contribute at most $n_{V_i}$. By Lemma \ref{la.lemma.7.5.5} we have that $W_A + W_{A,B} \leq 2\frac{\sqrt{\rho}}{\epsilon^2}W_{A, C}$. Combining the above yields
\begin{equation}
\label{la.eq.prop.8.2}
\begin{gathered}
\opt{G_i}{A_i} + \opt{G_i}{A_i,B_i} \leq 
n_{V_i}(W_{A_i} + W_{A_i, B_i}) \leq
n_{V_i} \cdot 2\frac{\sqrt{\rho}}{\epsilon^2}W_{A, C}.
\end{gathered}
\end{equation}

Finally, consider $\opt{G_i}{A_i,C_i}$.  By Corollary \ref{la.cor.1} applied to $\optsin{G_i}$ we have
\begin{equation}
\label{la.eq.prop.8.3}
\begin{gathered}
\opt{G_i}{A_i,C_i} \leq 
\frac12 n_{V_i} W_{A_i, C_i} (1 + \frac{9 \sqrt{\rho}}{\epsilon^2}).
\end{gathered}
\end{equation}

Thus, overall we get
\begin{equation*}
\begin{gathered}
OPT(G_i) = 
\opt{G_i}{A_i} + \opt{G_i}{A_i,B_i} + \opt{G_i}{A_i,C_i} + \opt{G_i}{B_i \cup C_i} \leq \\
n_{V_i} \cdot 2\frac{\sqrt{\rho}}{\epsilon^2}W_{A, C} + \frac12 n_{V_i} W_{A_i, C_i} (1 + \frac{9 \sqrt{\rho}}{\epsilon^2}) +  (1 + \frac{3 \sqrt{\rho_i}}{1 - \sqrt{\rho_i}}) \optsin{G_{i+1}} = \\
\frac12 n_{V_i} W_{A_i, C_i} (1 + \frac{13 \sqrt{\rho}}{\epsilon^2}) + (1 + \frac{3 \sqrt{\rho_i}}{1 - \sqrt{\rho_i}}) \optsin{G_{i+1}},
\end{gathered}
\end{equation*}
where the first equality follows from equality \ref{la.eq.prop.8.0} and the first inequality follows from inequalities \ref{la.eq.prop.8.1}, \ref{la.eq.prop.8.2} and \ref{la.eq.prop.8.3}. Finally, since $\rho_i < \epsilon^6$ and $\epsilon < \frac12$ we have $\frac{3 \sqrt{\rho_i}}{1 - \sqrt{\rho_i}} \leq 4 \sqrt{\rho_i}$ which in turn yields
\[
OPT(G_i) \leq
\frac12 n_{V_i} W_{A_i, C_i} (1 + \frac{13 \sqrt{\rho}}{\epsilon^2}) + (1 + 4 \sqrt{\rho_i}) \optsin{G_{i+1}}.
\]
Setting $\beta_i = \frac12 n_{V_i} W_{A_i, C_i} (1 + \frac{13 \sqrt{\rho}}{\epsilon^2})$ and $\gamma_i = 1 + 4 \sqrt{\rho_i}$ concludes the proof.
\end{proof}

\begin{proof}[Proof of Lemma \ref{la.lemma.7.5.5}]
We observe that trivially have that $W_A \leq \frac12 n_A^2 D_V$ and $W_{A,B} \leq n_A n_B D_V$ (recall that $D_V$ denotes the diameter of $V$). By Lemma \ref{general.lemma.1} we have that $n_A, n_B \leq \sqrt{\rho} n$ and that $n_C \geq (1 - \sqrt{\rho})n$ and therefore 
\begin{equation*}
\begin{gathered}
W_A + W_{A,B} \leq 
D_V n_A (\frac12 n_A + n_B)  \leq 
D_V n_A (1.5 \sqrt{\rho} n) \leq
D_V n_A n_C \cdot \frac{1.5 \sqrt{\rho}}{1 - \sqrt{\rho}}.
\end{gathered}
\end{equation*}
Finally, we note that by the definition of $C$ all weights between $A$ and $C$ are at least $\epsilon^2 D_V$ and therefore $W_{A,C} \geq \epsilon^2 n_A n_C D_V$. Combining this with the above inequalities yields
\begin{equation*}
\begin{gathered}
W_A + W_{A,B} \leq
\frac{1.5\sqrt{\rho}}{1 - \sqrt{\rho}} n_A n_C D_V \leq
\frac{1.5\sqrt{\rho}}{\epsilon^2(1 - \sqrt{\rho})} W_{A, C} \leq
2 \frac{\sqrt{\rho}}{\epsilon^2}W_{A, C},
\end{gathered}
\end{equation*}
where the third inequality follows since $\frac{1}{1 - \sqrt{\rho}} \leq \frac43$ (since $\rho < \epsilon^6$ and $\epsilon < \frac14$).
\end{proof}

\begin{proof}[Proof of Lemma \ref{la.lemma.11}]
Consider the algorithm $ALG_{d-w}$. It has a single loop that calls $PT(G, \Phi, \epsilon_{err} = \epsilon^9)$ and computes the value $\sum w_e \widehat{y}_e$ for the outputted partition. Computing the partition can be done in time $O(n^2)$. 

We consider Theorem \ref{general.thm.ggr} in order to bound the running time. We note that in our case the number of partition sets $k = \frac{1}{\epsilon}$. Therefore, the running time of $PT(G, \Phi, \epsilon_{err} = \epsilon^9)$ is bounded by
\[
\exp\big( \log(\frac{1}{\epsilon^9 }) \cdot (\frac{O(1)}{\epsilon^9})^{\frac{1}{\epsilon}+1}\big) + O(\frac{\log{\frac{1}{\epsilon^{10} }}}{\epsilon^{18}}) \cdot n.
\]

Consider the for loop within $ALG_{d-w}$. Every value $\beta_{j,j'}$ for a given pair $j, j'$ can have $\frac{1}{\epsilon^7}$ different values. Both $j$ and $j'$ can have $\frac{1}{\epsilon}$ values each. Therefore, the loop runs for $(\frac{1}{\epsilon^7})^{\frac{1}{\epsilon^2}}$ iterations. Therefore, the total running time of the algorithm is bounded by
\begin{equation*}
\begin{gathered}
(\frac{1}{\epsilon^7})^{\frac{1}{\epsilon^2}} \cdot \Big(\exp\big( \log(\frac{1}{\epsilon^9 }) \cdot (\frac{O(1)}{\epsilon^9})^{\frac{1}{\epsilon}+1}\big) + O(\frac{\log{\frac{1}{\epsilon^{10} }}}{\epsilon^{18}}) \cdot n + O(n^2)\Big) = 
(\frac{1}{\epsilon^7})^{\frac{1}{\epsilon^2}} \cdot O(n^2).
\end{gathered}
\end{equation*}

\end{proof}

\begin{proof}[Proof of Theorem \ref{la.thm.running_time}]
We first observe that case (b)'s running time is engulfed by that of case (a) and thus we may assume that the algorithm terminates in case (a).

Next we consider each recursion step and observe that its running time is defined by the time it takes to find $A_i$. In order to bound this running time consider the proof of Lemma \ref{general.lemma.1} and observe that it is algorithmic; one may iterate over all points and check for each point the amount of nodes of distance $\leq 2D_V \sqrt{\rho_V}$ - all in time $O(n^2)$ (which is linear in the size of the input). By Remark \ref{la.remark.loglog} we are guaranteed that the number of recursion steps is $O(  \log n)$ - summing to $O(n^2  \log n)$.

Therefore, together with Lemma \ref{la.lemma.11} (that bounds the running time of case (a)) we get that $ALG$ runs in time $O(n^2  \log n)$ plus the running time of $ALG_{d-w}$ (i.e., $f(\frac{1}{\epsilon}) \cdot O(n^2)$) which together yields an EPRAS.
\end{proof}

%%%%%%%%%%%%%%%%%%%%%%%%%%%%%%%%%%%%%%%%%%%%%%%%%%%%%%%%%%%%%%%%%%%%%%%%%%%%%%%%%%%%%%%%%%%%%%%%%%%%%%%%%%%%%%%%%%%
%%%%%%%%%%%%%%%%%%%%%%%%%%%%%%%%%%%%             HC SECTION START          %%%%%%%%%%%%%%%%%%%%%%%%%%%%%%%%%%%%%%%%
%%%%%%%%%%%%%%%%%%%%%%%%%%%%%%%%%%%%%%%%%%%%%%%%%%%%%%%%%%%%%%%%%%%%%%%%%%%%%%%%%%%%%%%%%%%%%%%%%%%%%%%%%%%%%%%%%%%

\section{Deferred Proofs of Section \ref{hc.sec}}

\begin{proof}[Proof of Proposition \ref{hc.proposition.2}]
Let $V$ denote the nodes of $G_k$ and $n_V = |V|$. Since we are in case (a) we have that $\rho_V \geq \epsilon^2$ and since $\rho_V = \frac{W_V}{n_V^2 D_V}$ we have that $\frac{W_V}{D_V} \geq n_V^2 \epsilon^2$. We argue that the instance has not-all-small-weights for $c_0 = c_1 = \epsilon^2$. Indeed, otherwise the total weight would be bounded by
\[
\frac{W_V}{D_V} < 
(1 - \epsilon^2) \cdot  \epsilon^2 \cdot  {n_V \choose 2} + \epsilon^2 \cdot 1 \cdot  {n_V \choose 2} <
2\epsilon^2 {n_V \choose 2} \leq
\epsilon^2 n_V^2,
\]
contradicting our assumption. Therefore, by Theorem \ref{hc.theorem.1} with $\xi = O(\epsilon^5)$ we are guaranteed that 
\[
\frac{ALG(G_k)}{OPT(G_k)} \geq
1 - O(\frac{\xi}{c_0 \cdot c_1}) = 
1 - \epsilon, 
\]
thereby concluding the proof.
\end{proof}

\begin{proof}[Proof of Proposition \ref{hc.prop.3}]
Recall that $ALG(G_k)$ is defined such that it first clusters $A$ as a ladder (denoted by $T_A$), then clusters $C$ arbitrarily (denoted by $T_C$) and finally roots $T_C$ to the bottom of $T_A$. Therefore, by the definition of the HC objective, every weight that is incident to $A$, adds to the objective its weight times $n_C$ and thus $ALG(G_k) \geq (W_A + W_{A,C})n_C$. By Lemma \ref{general.lemma.1} we have that $n_C \geq (1 - \sqrt{\rho})n$ and therefore $ALG(G_k) \geq (1 - \sqrt{\rho}) n (W_A + W_{A,C})$. 

Due to the fact that we are in case (b) we have that $W_C \leq 16 \epsilon W_V$ and therefore $W_A + W_{A,C} \geq (1 - 16 \epsilon )W_V$. Combined with what we explained above, we get $ALG(G_k) \geq (1 - \sqrt{\rho})(1 - 16 \epsilon ) n W_V$. Trivially, $OPT(G_k) \leq Wn$. Finally, again since we are in case (b) we have $\rho < \epsilon^2$. Overall,
\[
ALG(G_k) \geq 
(1 - \sqrt{\rho})(1 - 16 \epsilon ) OPT(G_k) \geq
(1 - \epsilon)(1 - 16 \epsilon )OPT(G_k) \geq
(1 - 17 \epsilon)OPT(G_k).
\]
\end{proof}

\begin{proof}[Proof of Lemma \ref{hc.lemma.4}]
Due to the fact that $ALG(G_i)$ places $A$ as a ladder and $C$ at the bottom of the ladder we have that $\alg{G_i}{A} + \alg{G_i}{A,C} \geq n_C(W_A + W_{A,C})$. On the other hand by Lemma \ref{general.lemma.1} we have that $n_C \geq n (1 - \sqrt{\rho})$. Finally, due to the fact that $ALG(G_i)$ recurses on $C$ we have that $\alg{G_i}{C} = ALG(G_{i+1})$. Therefore
\begin{equation*}
\begin{gathered}
ALG(G_i) = 
\alg{G_i}{A} + \alg{G_i}{A,C} + \alg{G_i}{C} \geq\\
n (W_A + W_{A,C})(1 - \sqrt{\rho}) + \alg{G_i}{C} = 
n (W_A + W_{A,C})(1 - \sqrt{\rho}) + ALG(G_{i+1}),
\end{gathered}
\end{equation*}
thereby concluding the proof.
\end{proof}

\begin{proof}[Proof of Lemma \ref{hc.lemma.6}]
We first observe that 
\begin{equation*}
\begin{gathered}
OPT(G_i) = 
\opt{G_i}{A_i} + \opt{G_i}{A_i,C_i} + \opt{G_i}{C_i} \leq
n_{V_i}(W_{A_i} + W_{A_i, C_i}) +  \opt{G_i}{C_i},
\end{gathered}
\end{equation*}
since every edge may contribute at most $n_{V_i}$. Consider $\opt{G_i}{C_i} = \sum_{e \in C_i} w_e T_e$. The value $T_e$ is comprised of nodes from $n_{A_i}$ and $n_{C_i}$. Therefore $\opt{G_i}{C_i} \leq n_{A_i} W_{C_i} + \optsin{G_{i+1}}$, since $\optsin{G_{i+1}}$ solves the instance defined by $C_i$ optimally. Therefore $OPT(G_i) \leq n_{V_i}(W_{A_i} + W_{A_i, C_i}) + n_{A_i} W_{C_i} + \optsin{G_{i+1}}$.

By Fact \ref{fact.avg_link} we have that $\optsin{G_{i+1}} \geq \frac23 n_{C_i} W_{C_i}$. Additionally, by Lemma \ref{general.lemma.1}, we have $n_{C_i} \geq \frac{1 - \sqrt{\rho_i}}{\sqrt{\rho_i}} \cdot n_{A_i}$. Combining these inequalities with the above yields
\begin{equation*}
\begin{gathered}
OPT(G_i) \leq 
n_{V_i}(W_{A_i} + W_{A_i, C_i}) + n_{A_i} W_{C_i} + \optsin{G_{i+1}} \leq \\
n_{V_i}(W_{A_i} + W_{A_i, C_i}) + \frac{\sqrt{\rho_i}}{1 - \sqrt{\rho_i}} n_c W_C + \optsin{G_{i+1}} \leq \\
n_{V_i}(W_{A_i} + W_{A_i, C_i}) + (1 + \frac{\frac32 \sqrt{\rho_i}}{1 - \sqrt{\rho_i}}) \optsin{G_{i+1}}.
\end{gathered}
\end{equation*}
Finally, since $\rho_i < \epsilon^2$ and $\epsilon < \frac12$ we have $\frac{\frac32 \sqrt{\rho_i}}{1 - \sqrt{\rho_i}} \leq 2 \sqrt{\rho_i}$ and therefore 
\[
OPT(G_i) \leq \beta_i + \gamma_i \optsin{G_{i+1}}
\]
for $\beta_i = n_{V_i}(W_{A_i} + W_{A_i, C_i})$ and $\gamma_i = 1 + 2 \sqrt{\rho_i}$.
\end{proof}

To prove Theorem \ref{hc.thm.last}, we will need to show that $\Pi_{j=0}^{i-1} \gamma_j$ converges. Fortunately, the weighted densities $\rho_i$ increase fast enough to ensure this. All proofs of this subsection are deferred to the Appendix.

\begin{lemma}
\label{hc.lemma.7}
For all $i = 0, 1, \ldots, k-1$ we are guaranteed that $\rho_{i+1} \geq 4 \rho_i$.
\end{lemma}

\begin{proof}[Proof of Lemma \ref{hc.lemma.7}]
First observe that $\rho_{i+1} \geq 4 \epsilon \sqrt{\rho_i}$, which follows from the fact that $W_{V_{i+1}} \geq 16 \epsilon W_{V_i}$, $n_{V_{i+1}} \leq n_{V_i}$ and $D_{V_{i+1}} \leq 4D_{V_{i}}\sqrt{\rho_i}$, which all follow from Lemma \ref{general.lemma.1} and since case (c) applies. Next, also due to the fact that case (c) applies, we have that $\rho_i \leq \epsilon^2$ and therefore $4 \epsilon \sqrt{\rho_i} \geq 4 \rho_i$, thereby concluding the proof.
\end{proof}

\noindent We are now ready to show that $\Pi_{j=0}^{i-1} \gamma_j$ converges. 

\begin{corollary}
\label{hc.cor.7.5}
For $\gamma_i =  1 + 2 \sqrt{\rho_i}$ we have $\Pi_{j=0}^{i-1} \gamma_j \leq 1 + 3\sqrt{\rho_i}$.
\end{corollary}

\begin{proof}[Proof of Corollary \ref{hc.cor.7.5}]
Observe that 
\[
\Pi_{j=0}^{i-1} (1 + 2\sqrt{\rho_j}) \leq
e^{2 \cdot \sum_j \sqrt{\rho_j}} \leq
e^{2 \sqrt{\rho_i}} \leq
1 + 3 \sqrt{\rho_i},
\]
where the first inequality follows from Observation \ref{general.ob.0}, the second follows since $\sqrt{\rho_j}$ are exponentially increasing (Lemma \ref{hc.lemma.7}) and the third inequality follows again by Observation \ref{general.ob.0} combined with the fact that $\rho < \epsilon^2$ and $\epsilon < 10^{-2}$.
\end{proof}

\noindent Next we leverage the former lemmas to bound $\min_i \{\frac{\alpha_i}{\beta_i \Pi_{j=0}^{i-1} \gamma_j}\}$ and $\frac{ALG(G_k)}{(\Pi_{i=0}^{k-1} \gamma_i)OPT(G_k)}$.

\begin{proof}[Proof of Proposition \ref{hc.prop.11}]
By the definitions of $\alpha_i$, $\beta_i$ and $\gamma_i$ we have
\begin{equation*}
\begin{gathered}
\frac{\alpha_i}{\beta_i \Pi_{j=0}^{i-1} \gamma_j} = 
\frac{1 - \sqrt{\rho_i}}{ \Pi_{j=0}^{i-1} \gamma_j} \geq 
\frac{1 - \sqrt{\rho_i}}{1 + 3 \sqrt{\rho_i}} \geq
1 - 4 \sqrt{\rho_i},
\end{gathered}
\end{equation*}
where the first equality is due to the definitions of $\alpha_i$ and $\beta_i$ and the first inequality is due to Corollary \ref{hc.cor.7.5} and the definition of $\gamma_i$. Therefore, since $\rho_i$'s only increase, 
\[
\min_i \{\frac{\alpha_i}{\beta_i \Pi_{j=0}^{i-1} \gamma_j}\} \geq
1 - 4 \sqrt{\rho_{k-1}} \geq
1 - 4 \epsilon,
\]
where the last inequality follows since $\rho_{k-1} < \epsilon^2$ - thereby concluding the proof.
\end{proof}

\begin{proof}[Proof of Proposition \ref{hc.prop.12}]
If $k=0$ then we have no recursion and by the proof holds by Propositions \ref{hc.proposition.2} and \ref{hc.prop.3}. Otherwise 
\[
\frac{ALG(G_k)}{(\Pi_{i=0}^{k-1} \gamma_i)OPT(G_k)} \geq
\frac{1 - 17 \epsilon}{(1 + 2 \sqrt{\rho_{k-1}})(1 + 3 \sqrt{\rho_{k-1}})} \geq
1 - 23 \epsilon,
\]
where the first inequality follows from Corollary \ref{hc.cor.7.5} and the definition of $\gamma_{k-1}$ and thesecond inequality follows since $\rho_{k - 1} < \epsilon^2$ (since we recursed to step $k$) - thereby concluding the proof.
\end{proof}

\begin{proof}[Proof of Theorem \ref{hc.thm.last}]
Follows from Observation \ref{general.ob.1} and Propositions \ref{hc.prop.11} and \ref{hc.prop.12}.
\end{proof}

\begin{lemma}
\label{hc.lemma.loglog}
The number of recursion steps performed by Algorithm \ref{alg.hc_alg} is bounded by $O(\log \log n)$.
\end{lemma}

\begin{proof}[Proof of Lemma \ref{hc.lemma.loglog}]
Let $\rho_0$ denote the density of the original graph $G_0 = (V_0,W_0)$ and let $n_0 = |V_0|$. By the triangle inequality we have $W_0 \geq D_{V_0}(n_0 - 1)$. We can see this by considering $u,v$ with $w_{u,v} = D_{V_0}$ and any $w \neq u,v$. By the triangle inequality $w_{u,w} + w_{v,w} \geq D_{V_0}$. Summing over all $w$ and adding this to $w_{u,v}$ results in $W_0 \geq D_{V_0}(n_0 - 1)$. Therefore $\rho_0 = \frac{W_0}{n_0^2 D_{V_0}} \geq \frac{1}{n_0} - \frac{1}{n_0^2} \geq \frac{1}{2n_0}$. 

As a bi-product of the proof of Lemma \ref{hc.lemma.7} we are guaranteed that $\rho_{i+1} \geq 4 \epsilon \sqrt{\rho_i}$. Therefore $\rho_0 \leq (\frac{1}{16 \epsilon^2})^{2^{k}-1} \cdot \rho_{k-1}^{2^k}$. On the other hand, by the definition of our algorithm, if we performed a recursion step then $\rho_i \leq \epsilon^2$. Thus, if we consider $k-1$ as the last recursion step, we have that $\rho_{k-1} \leq \epsilon^2$. 

Combining all of the above yields 
\[
\frac{1}{2n_0} \leq 
\rho_0 \leq 
(\frac{1}{16 \epsilon^2})^{2^{k}-1} \cdot \rho_{k-1}^{2^k} \leq 
(\frac{1}{16 \epsilon^2})^{2^{k}-1} \cdot (\epsilon^2)^{2^k}.
\] 
Extracting $k$ yields $k = O(\log \log n)$. 
\end{proof}

\begin{proof}[Proof of Theorem \ref{hc.thm.running_time}]
We first observe that case (b)'s running time is engulfed by that of case (a) and thus we may assume that the algorithm terminates in case (a).

Next we consider each recursion step and observe that its running time is defined by the time it takes to find $A_i$. In order to bound this running time consider the proof of Lemma \ref{general.lemma.1} and observe that it is algorithmic; one may iterate over all points and check for each point the amount of nodes of distance $\leq 2D_V \sqrt{\rho_V}$ - all in time $O(n^2)$ (which is linear in the input size). By Lemma \ref{hc.lemma.loglog} we are guaranteed that the number of recursion steps is $O( \log \log n)$ - summing to $O(n^2 \log \log n)$.

Therefore, together with Theorem \ref{hc.theorem.1} (that bounds the running time of case (a)) we get that $ALG$ runs in time $O(n^2  \log \log n)$ plus the running time of $ALG_{d-w}$ (i.e., $f(\frac{1}{\epsilon^5}) \cdot O(n^2)$) which together yields an EPRAS.
\end{proof}

\end{document}